\documentclass[11pt,a4paper]{article}

\usepackage[T1]{fontenc}
\usepackage[utf8]{inputenc}
\usepackage{lmodern}
\usepackage[margin=1in]{geometry}
\usepackage{setspace}

\usepackage{amsmath,amssymb,amsfonts,amsthm}
\usepackage{mathrsfs}
\usepackage[mathscr]{eucal}
\usepackage{dsfont}
\usepackage{bm}
\allowdisplaybreaks

\usepackage{graphicx}
\usepackage{booktabs}
\usepackage{array,multirow,makecell}
\usepackage{enumitem}
\usepackage{algorithm}
\usepackage{caption}
\usepackage{subcaption}
\usepackage{siunitx}
\usepackage[square,numbers]{natbib}

\usepackage{authblk}

\usepackage[colorlinks, allcolors=blue]{hyperref}
\usepackage{xurl}

\theoremstyle{plain}

\newtheorem{proposition}{Proposition}

\theoremstyle{definition}

\theoremstyle{remark}
\newtheorem{remark}{Remark}

\numberwithin{equation}{section}

\usepackage[most]{tcolorbox}
\newtcolorbox{sidebar}[1][]{breakable, enhanced, colback=black!3, colframe=black!60, boxrule=0.4pt, left=6pt, right=6pt, top=6pt, bottom=6pt, #1}

\newcommand{\id}{\mathrm{id}}

\newcommand{\R}{\mathbb{R}}
\newcommand{\N}{\mathbb{N}}

\newcommand{\E}{\mathbb{E}}

\newcommand{\argmin}{\operatornamewithlimits{arg\,min}}

\newcommand{\verti}[1]{{\left\vert\kern-0.25ex\left\vert\kern-0.25ex\left\vert #1 \right\vert\kern-0.25ex\right\vert\kern-0.25ex\right\vert}}
\newcommand{\vertj}{\vert\kern-0.25ex\vert\kern-0.25ex\vert}
\newcommand{\ip}[1]{{\left\langle\kern-0.5ex\left\langle\kern-0.5ex\left\langle #1 \right\rangle\kern-0.5ex\right\rangle\kern-0.5ex\right\rangle}}

\begin{document}

\title{Amplitude-Phase Analysis of the COVID-19 Point Process\\ and the Early Countermeasures}

\author[1]{Francesco Tripoli}
\author[2]{Leonardo V. Santoro}
\author[3]{Tomas Masak}
\affil[1]{Harvard Business School, Boston, Massachusetts, United States.}
\affil[1]{\texttt{ftripoli@hbs.edu}}
\affil[2]{Institut de Math\'ematiques, \'Ecole Polytechnique F\'ed\'erale de Lausanne, Switzerland.}
\affil[2]{\texttt{leonardo.santoro@epfl.ch}}
\affil[3]{Institute for Statistics and Mathematics, Wirtschaftsuniversit\"{a}t Wien, Austria.}
\affil[3]{\texttt{tomas.masak@wu.ac.at}}

\date{\today}

\maketitle

\begin{abstract}
\noindent We investigate how governmental restrictions relate to the spread and temporal dynamics of COVID-19 early in the pandemic. We model daily infection data from each US state as realisations of a point process, taking the random intensity measure to be the latent object of interest and, crucially, allowing these realisations to vary not only in magnitude but also in the temporal dynamics. By non-parametrically separating these amplitude and phase variations, we examine how government restrictions relate to each source of variability, relating the infection curves to the Oxford Stringency Index, which we treat as a measure on the same time window. Employing Wasserstein PCA, we analyse the temporal variability of both the infections and the restrictions. We then use the resulting scores, together with the scalars representing the overall stringency budget and the total infection count, as inputs to a linear vector-on-vector regression model. Our findings suggest that, when considering the separate contributions of amplitude and phase variability, earlier implementation of restrictions is associated with flatter infection curves. By contrast, we do not find significant evidence of an association between stringency and total infection counts, nor between the overall stringency budget and the infection outcomes considered.


\end{abstract}

\section{Introduction}
In March 2020, about three months after initial reports on the spread of the novel SARS-CoV-2 virus in the area of Wuhan, China, governments worldwide began implementing restrictions on citizen mobility and social interaction, aimed at mitigating the surge in cases and flattening the pandemic curve.
Indeed, the sudden and swift progression of events led to initial pandemic responses being primarily informed by real-time updates on infection trends. Public events were cancelled, and mobility restrictions, including stay-at-home orders, were enforced. Extensive studies \citep[e.g.][]{Guedhami2022, Arendt2020} have delved into the profound economic and psychological effects of these measures, underscoring the need for careful consideration when implementing such regulations due to their socio-economic impacts.

While the economic implications of these restrictions are now better understood, the actual effects of government measures on shaping the pandemic's trajectory remain unclear. Notably, the OECD's 2022 report on the pandemic has highlighted the necessity for further research into the efficacy of mobility restrictions \citep{OECD}.

We attempt to address the OECD call by employing methodologies from the areas of functional data analysis \citep{ramsay2005,hsing2015} and statistical optimal transport \citep{Santambrogio2015,Panaretos2020}. Functional data analysis concerns the inference of a random process given multiple realisations thereof, while statistical optimal transport concerns the inference on probability distributions seen themselves as latent data objects.
We consider the daily infection counts in different US states early in the pandemic as a random sample of point processes, and we use the methodology of \citet{Panaretos2016} to separate the amplitude (i.e.~the \textit{intensity}) and phase variation (i.e.~the \textit{time dynamics}) in the COVID-19 daily infection counts; see Figure~\ref{fig:Figure_1} for a visual illustration, and Section~\ref{sec:preliminaries} for a brief introduction to amplitude and phase variation in functional data.

We regard the daily infection counts of each US state during the first wave as a~realisation of a point process, and the corresponding stringency index as a measure on the same time window. Each of the two objects is then described by its total mass -- the cumulative case count of a state, and the state's overall ``stringency budget'' -- together with a probability density encoding \emph{when} that mass was placed. The temporal components are compared through the lens of optimal transport: amplitude and phase variation are separated by a canonical registration procedure which leverages the point-process modelling. The temporal dynamics of the pandemic, i.e.~when the cases in a given state arrived and how concentrated in time they were, can be meaningfully compared across states and related to the timing of restrictions only after the two layers of variability have been disentangled. Our main insights arise from this separation and from treating the resulting phase component as an object of primary interest.

\begin{figure}[t!]
    \centering
    {\includegraphics[width=15.5cm]{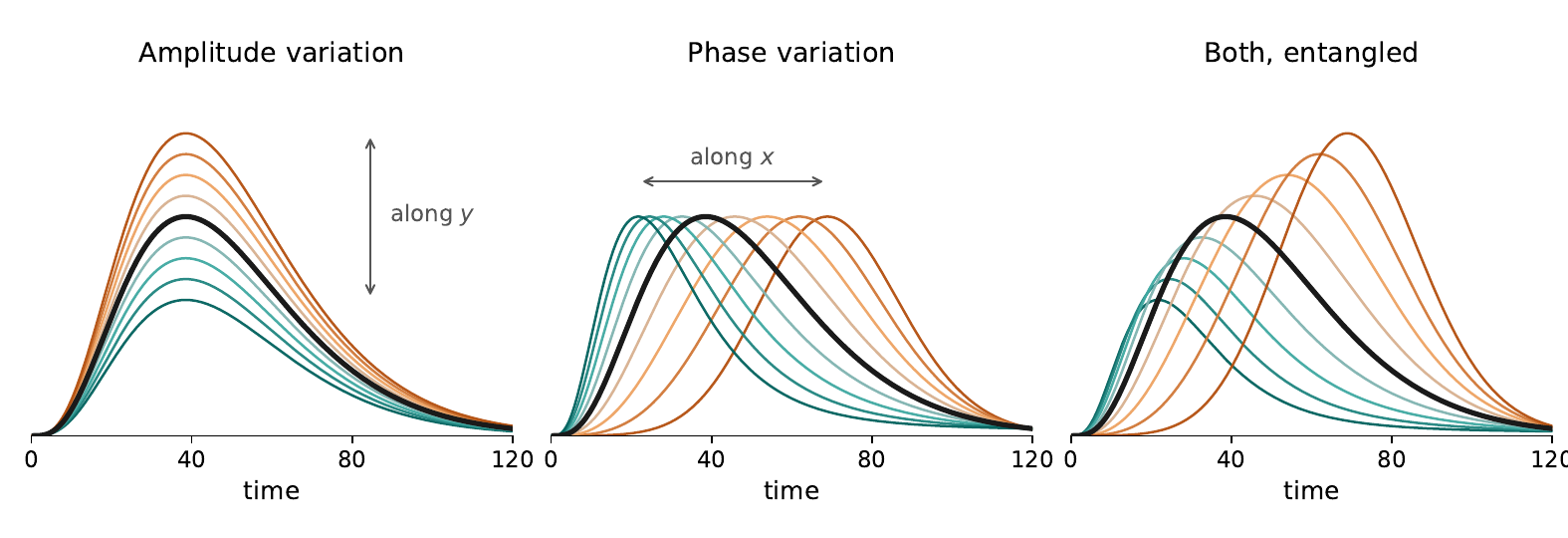} }
    \caption{Amplitude and Phase variation in functional data: samples obtained by additive perturbation of the mean (left), time-warping of the mean (center), and both effects entangled (right).}
    \label{fig:Figure_1}
\end{figure}

Such an emphasis is at odds with the prevailing view. The separation of amplitude and phase variation constitutes an active area of research that has been growing fast in recent years, especially in the field of functional data analysis \citep[see, for example,][and the references therein]{kneip1992statistical, kneip2000curve, kneip2008combining, srivastava2011registration, marron2015functional, zhang2015elastic, chakraborty2021functional}, yet the phase variation is typically considered a nuisance that needs to be taken care of in order not to distort the findings pertinent to amplitude variation, which is usually regarded as the primary object of interest. For example, in the abstract of their comprehensive review paper, \citet{marron2015functional} write: \emph{``The presence of phase variability artificially often inflates data variance, blurs underlying data structures, and distorts principal components.''} In contrast to this prevailing perspective, we argue that phase variation can itself be of primary interest. In the context of COVID-19 infection curves, temporal dynamics encode essential information about the timing and progression of the pandemic across regions -- information that registration would ordinarily discard. Recently, \citet{aston2025} advocated this viewpoint in their analysis of COVID-19 incidence in the United Kingdom.

\citet{marron2015functional} also note that the separation problem is difficult and only admits a canonical solution under stringent assumptions. From the practical point of view, they recommend tailoring registration algorithms to the problem area. Nonetheless, when the data can be viewed as (potentially time-distorted and scaled) realisations of a point process, the separation problem admits a canonical solution \citep{Panaretos2016}. Furthermore, as observed by \citet{gajardo2023point}, this provides a
``\textit{natural framework in this setting, since confirmed infections (cases) and deaths due to COVID-19 are events which arrive at random times within each region of interest, i.e.~they come from a temporal point process mechanism for each region''}, and many approaches in the literature account for the counting-process nature of COVID-19 data
\citep{dong2023non,lee2020estimating,li2021understanding,chen2022novel}.

Our perspective on the infection curves has the advantage that it comes with a canonical solution to the separation problem, put forward by \citet{Panaretos2016}.
This also allows us to model the phase variability in the COVID case count process using Wasserstein tangent-space principal component analysis (WTPCA) \citep[cf.][]{petersen2016functional,Panaretos2020,bigot2026pca}. In the spirit of functional regression regularised via functional principal component analysis (FPCA), we subsequently utilise the principal component (PC) scores in a vector-on-vector regression model to finally relate the infection curves to restrictions: the main covariate in our analysis is the Oxford stringency index \citep{Hale2021}, which combines a variety of individual governmental restrictions in a time-varying scalar value.
We view this index through the same distributional lens, and perform WTPCA after scaling.
The final regression model can thus be seen as a Wasserstein tangent-space version of similar distribution-on-distribution regression models
\citep{ghodrati2022,chen2023wasserstein,chen2024distribution}.
While we use the COVID application (specifically the Oxford stringency index as an input and the COVID case counts as the output, together with additional control covariates) as a relatable case study, we believe that utilising WTPCA in standard regression analysis is an interesting option for many applications where distributional variables are commonly summarised using several moments and/or quantiles. Examples include age(-at-death) distribution \citep{bergeron2019age,camarda2024bayesian}, (vaccination timing) exposure distribution \citep{dey2016vaccine,fox2024postexposure}, or wealth distribution \citep{davies2000distribution,blanchet2022uncovering}, just to name a few.

The main contribution of this paper is thus two-fold. Utilising recently developed techniques from functional data analysis and optimal transport, we put forward a fresh perspective on the role of restrictions in the COVID-19 pandemic, one that is particularly light on modelling assumptions. Moreover, we propose a Wasserstein tangent space distribution-on-distribution regression model with external covariates that is highly interpretable, and we highlight this using the well-studied and known COVID-19 data set (specifically the first wave of the pandemic in the 50 US states) as a case study. In our opinion, this is a rare instance of phase variability in functional data being of an equal, or possibly even primary importance, compared to variation in amplitude. From the applied perspective, our analysis finds early restrictions to be associated with flatter infection curves, thus contributing to the OECD's call for further investigations on the efficacy of such restrictions.

The remainder of this paper is structured as follows. Section~\ref{sec:data} introduces the data used in this case study, the overall statistical framework, and the necessary pre-processing steps. Section~\ref{sec:methods} provides an overview of the methodology employed in this study. The results are then presented in Section~\ref{sec:results}, followed by a discussion of our findings and potential deficiencies. All data analyses can be reproduced straightforwardly using the accompanying repository. Further methodological details, diagnostics, and proofs are provided in the appendices.

\section{Data and Setup} \label{sec:data}

To provide a clear context for our study, we begin by describing the data sources and the pre-processing steps employed in this work.

The data on daily counts of COVID-19 infections in the United States are sourced at a sub-national level from the New York Times COVID-19 repository~\citep{NYTimes2021}. We intentionally focused on the US in order to ensure a more consistent and homogeneous approach to data collection and reporting for each state. Daily counts were obtained by first differencing the cumulative series and rescaled using the 2019 resident population from the U.S. Census Bureau estimates \citep{USCensus2020}. Furthermore, in order to track governments' responses, we gathered data from the Oxford Stringency Index \citep{Hale2021}, a comprehensive metric averaging nine policy indicators that encompass various governmental measures to slow down the spread of the pandemic, including -- for example -- school and workplace shutdowns or travel restrictions. The index ranges from 0 to 100, with the latter signifying the most stringent measures. Figure~\ref{fig:raw_data} displays the cumulative infection counts and the aggregated Oxford stringency index for individual states. Recall that our goal is to investigate the association between the shape of the infection curves (shown in Figure~\ref{fig:raw_data}, left) and the shape of the stringency curves (shown in Figure~\ref{fig:raw_data}, right), treating the former as the response and the latter as the main predictor.

\begin{figure}[t!]
  \centering
   \begin{tabular}{cc}
           \includegraphics[width=0.48\textwidth]{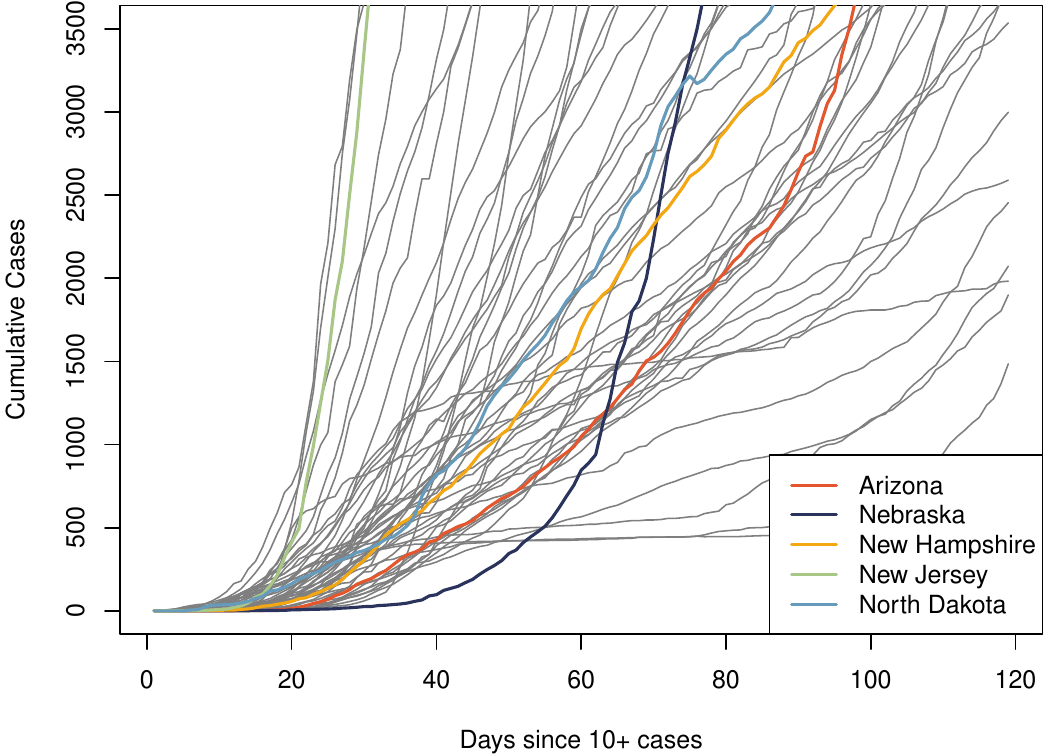}  &
           \includegraphics[width=0.48\textwidth]{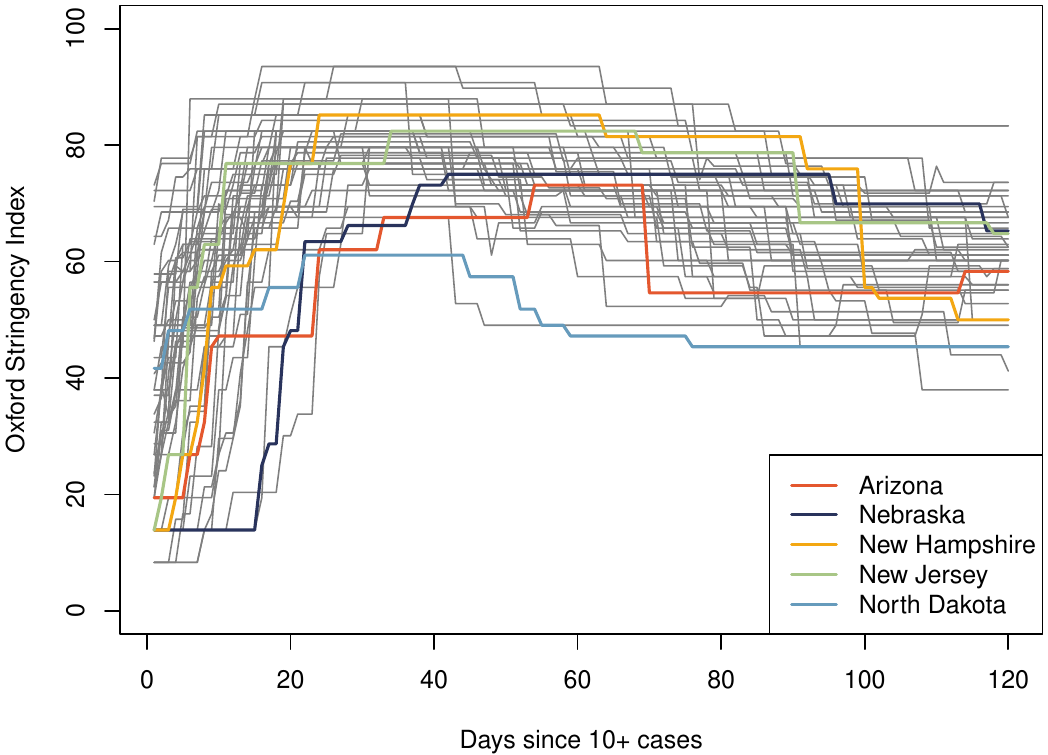}
        \end{tabular}
  \caption{Cumulative infection counts per million inhabitants (left) and Oxford Stringency Index (right) for the fifty US states since 21 days before the time when the cumulative infection counts reached 10 per million inhabitants in a given state. Five states are highlighted in colour for easy comparison between the figures.}
  \label{fig:raw_data}
\end{figure}

For the purposes of this analysis, we treat these data as functional since they represent the evolution of a process that changes continuously over time. The data undergo an initial smoothing procedure, which is both crucial for the methodology (cf.~Remark~\ref{rem:imperfect_obs} below) and well-suited for this particular application.
Smoothing the case counts is desirable from the applied perspective here, since especially during the initial phases of the pandemic, cases might be recorded with a slight delay or, conversely, there could be an initial over-reporting of cases on day $t$, later corrected by marking fewer cases on day $t+1$. Similarly, smoothing the stringency index is plausible considering that the propagation of a specific mitigation measure through the population is not immediate. For example, some may choose to remain in their houses before a formal stay-at-home order is issued, when such an order is expected, while others might be reached by the order only some time after its enforcement.

The pre-processing is carried out on state-specific first-wave windows. We assumed that the first wave of the pandemic endured for 120 days from its beginning, defined by 21 days before the time when the cumulative infection count reached 10 per million inhabitants in a~given state. A similar strategy was used by \citet{carol2020}.
Each curve of daily infection counts is then smoothed with a Gaussian Nadaraya--Watson kernel smoother with bandwidth $h=14$ days. For the stringency index, the daily index values are treated as a discrete measure on the same time grid and smoothed with the local polynomial density estimator of \citet{Cattaneo2020}, implemented through \texttt{lpdensity}, using a local quadratic fit, the default triangular kernel, and a bandwidth of $h=14$ days. To alleviate boundary issues, three additional days are included on each side of the target window before evaluating the smoothed curve on the 120-day grid. While conventional density estimators typically involve smoothing a histogram-based density estimator, \citet{Cattaneo2020} propose to smooth the empirical distribution function using local polynomial techniques. Consequently, their density estimator relies on a preliminary $\sqrt{N}$-consistent distribution function estimator. This approach requires only the choice of the bandwidth parameter associated with the local polynomial fit at each evaluation point. This method is preferred to the traditional kernel density estimator near the boundaries of the analysis window, where the latter may require additional corrections \citep{Wand1995, Cattaneo2020}. The data are subsequently evaluated across a discrete, evenly spaced grid between the initial time, $t_i$ (here the onset of the first wave of the pandemic), and the final time $t_f$ (i.e.~the moment marking its conclusion). Figure~\ref{fig:time_windows} shows the time periods considered for each individual state.

An initial data visualisation reveals that the case count curves exhibit two, or often three, periods of accelerated growth. The beginning of these phases varies slightly across states and signifies the occurrence of the three ``waves'' of the pandemic between March 2020 and June 2021. We focused on the first wave of COVID-19, due to the variability in the number of waves experienced by different states and because we conceptualised a single pandemic wave as a random point process, as detailed above.

\begin{figure}[t!]
  \centering
  \includegraphics[width=0.8\linewidth]{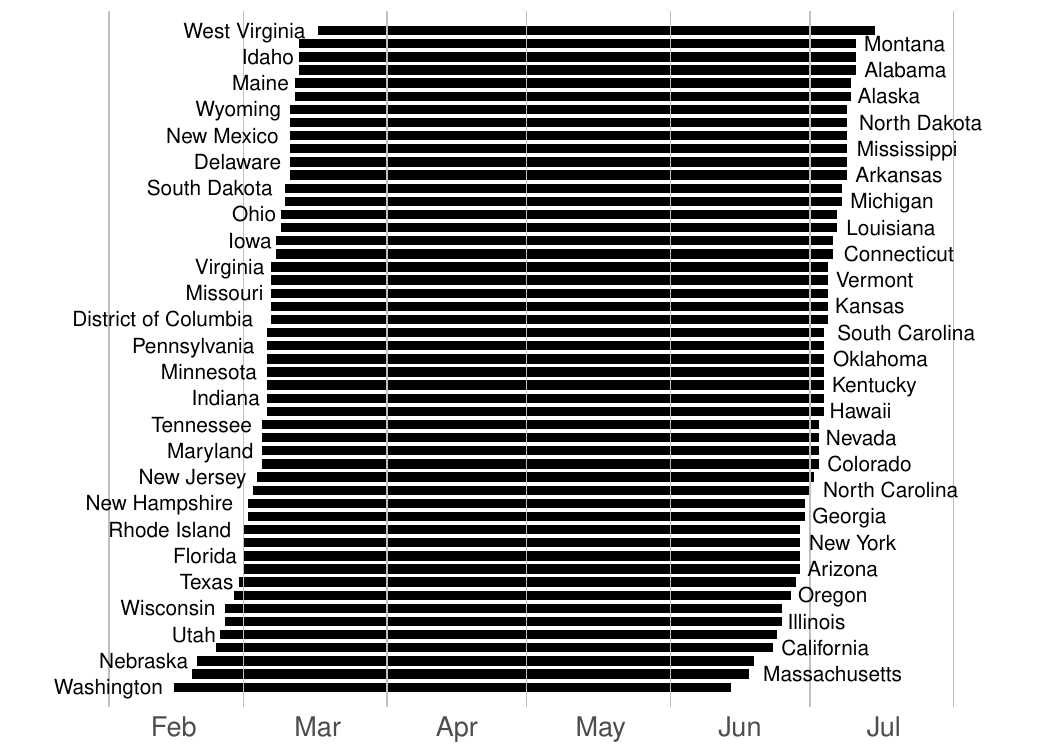}
  \caption{Time periods pertinent to the first wave of COVID-19 for each US state.}
  \label{fig:time_windows}
\end{figure}

As the last step of the pre-processing pipeline, we expressed each state's distribution data (i.e.~the scaled daily infection counts and stringency index curves) as densities. This involves removing constants, which are then used as regression inputs. Specifically, for each state we write
\[
c_n=\sum_t \widehat Y_n(t),
\qquad
f_n(t)=\frac{\widehat Y_n(t)}{c_n},
\qquad
b_n=\sum_t \widehat X_n(t),
\qquad
g_n(t)=\frac{\widehat X_n(t)}{b_n},
\]
where $\widehat Y_n$ and $\widehat X_n$ denote the smoothed infection and stringency curves, respectively. This step is desirable since, as detailed in Sections~\ref{sec:registration} and~\ref{sec:WTPCA}, both the registration procedure and the Wasserstein PCA operate at the density level. We then attach specific interpretations to each of the constants: the ones associated with infection counts, $c_n$, reflect the overall impact of COVID-19 on each state -- the higher the constant, the greater the number of infections recorded. Conversely, the constants related to the stringency index, $b_n$, represent the ``stringency budget'' for each state. Essentially, this can be seen as each state being able to afford enforcing a certain overall level of stringency, beyond which the economic and social consequences would be deemed exorbitant. Each governor then decides how to allocate this budget across different days of the pandemic and in what proportions. Considering the stringency curves in Figure~\ref{fig:raw_data} (right) as instances of scaled densities (to be smoothed) is a particular viewpoint adopted in the data analysis.

Finally, the control variables used in regression analysis are the state-level GDP per capita, defined as the 2019 annual current-dollar GDP from the Bureau of Economic Analysis \citep{BEA2020GDP} divided by the 2019 resident population \citep{USCensus2020}; and the population density for each state, computed from the 2019 resident population and state land area. Both of these variables were log-transformed. We also considered the percentage of democratic voters, measured as the Democratic vote share in the latest statewide election preceding the pandemic as sourced from The Council of State Governments \citep{CSG2020}. However, this variable was ultimately dropped as unimportant.

\section{Methods}\label{sec:methods}

Recall from Section~\ref{sec:data} that we regard the infection curves of the individual states as i.i.d.\ realisations of a point process, distorted both in amplitude and in phase, and that both these curves and the stringency index have been reduced to a total mass and a probability density. This section develops the machinery needed to relate the two, all of it applied to the smoothed data of Section~\ref{sec:data}. We first collect the necessary background, on amplitude and phase variation and on the geometry of the space of measures. We then describe, in turn, the registration procedure that separates the two sources of variability, the Wasserstein tangent-space PCA by which the temporal dynamics are summarised in a few scores, and the vector-on-vector regression through which those scores are finally related.

\subsection{Preliminaries}\label{sec:preliminaries}
\paragraph{Amplitude and phase variation.}
One refers to amplitude and phase variation to identify two different types of variation of a random function $f\::\:[0,T]\to \mathbb{R}$ over a continuous compact domain $[0, T]$, cf.~Figure~\ref{fig:Figure_1}. The first type of variation consists of random fluctuations around a~mean level, for instance arising as:
$$f(t) = \mu(t) + \varepsilon(t),$$
where $\mu(t)= \mathbb{E}[ f(t) ]$ denotes the mean of the random function $f$ and $\varepsilon$ is a zero-mean perturbation, typically assumed to be smooth in some sense. Thus, the amplitude variation refers to a fluctuation ``along the y-axis''.
On the other hand, phase variation is non-linear and arises from random deformations in the time scale, speeding up or slowing down the process evolution, and yielding time-warped curves:
$$
\tilde{f}(t) = \mu(\varphi^{-1}(t)) + \varepsilon(\varphi^{-1}(t)) ,
$$
where $\varphi: [0,T] \rightarrow [0,T]$ is typically referred to as the \textit{warping function}. This is a random increasing function, independent of $\varepsilon$, preserving the time scale on average, i.e.~$\mathbb{E}[\varphi(t)] = t$. Thus, phase variation induces a fluctuation ``along the x-axis''.

Phase variation arises naturally in random processes that lack an absolute notion of time, and each realisation evolves according to a speed that is intrinsic to the process itself. This is for instance the case of growth curves, as in \citet{ramsay2005,ramsay2002applied}, but many other examples can be found in the literature, such as heartbeat signals, speech waveforms, or activity profiles, just to name a few.
The entanglement of these two layers of variation often occurs intrinsically to many processes. It is fundamental to correctly separate and account for the contribution of each layer in order to deduce valuable information on the process, and failing to do so may seriously distort the outcomes of an analysis \citep{marron2015functional}.

\vspace{-3mm}
\paragraph{The Wasserstein space.}  Let $\mathcal{X}\subset\mathbb{R}$ be compact and consider the space of second-order probability measures on $\mathcal{X}$,
$$
\mathcal{W}_2(\mathcal{X})=\left\{\nu \in \mathcal{P}(\mathcal{X}): \int_{\mathcal{X}}|x|^2 \mathrm{d}\nu(x)<\infty\right\},
$$
where $\mathcal{P}(\mathcal{X})$ denotes the probability measures on $\mathcal{X}$. The 2-Wasserstein distance between $\mu,\nu \in \mathcal{W}_2(\mathcal{X})$ is given by the Monge problem of optimal transportation,
$$
W^2_2(\mu,\nu):= \inf \left\{\int_{\mathcal{X}}\left\{T(u)-u\right\}^2 \mathrm{d}\mu(u) \ \text{ such that } \ T_{\#}\mu=\nu\right\},
$$
where $T_{\#}\mu$ denotes the push-forward of $\mu$ through a $\mu$-measurable map $T$. For absolutely continuous measures this problem has a unique solution, induced by the \textit{optimal} transport map, which is available in closed form as
$
T_{\mu}^{\nu} := F_\nu^{-1} \circ F_\mu,
$
where $F_\mu$ denotes the distribution function of $\mu$ and $F_\mu^{-1}$ its right-continuous quantile function. Substituting this map into the Monge problem yields the equally explicit expression
\begin{equation}\label{eq:w2-quantile}
W_2^2(\mu,\nu) = \int_0^1 \left\{F^{-1}_{\mu}(u) - F^{-1}_{\nu}(u)\right\}^2 \mathrm{d}u,
\end{equation}
so that the quantile representation embeds $\mathcal{W}_2(\mathcal{X})$ isometrically into $\mathcal{L}^2(0,1)$, its image being the closed convex set of non-decreasing functions therein. It is in this sense that the univariate Wasserstein space is \emph{flat}, in contrast with its higher-dimensional counterparts, and this flatness underlies and enables most of what follows.

The Wasserstein space $\left(\mathcal{W}_2(\mathcal{X}), W_2\right)$ may be endowed with a formal Riemannian structure: at each absolutely continuous measure $\mu$ there exists a flat space, together with a projection onto it, that faithfully describes small perturbations around $\mu$. Concretely, writing $\mathcal{L}^2(\mu)$ for the $\mu$-square-integrable real functions on $\mathcal{X}$ equipped with the inner product
$$
\left\langle v_1, v_2\right\rangle_\mu=\int_{\mathcal{X}} v_1(u) v_2(u)\, \mathrm{d}\mu(u), \qquad v_1, v_2 \in \mathcal{L}^2(\mu),
$$
the set $\left\{\gamma ( T_{\mu}^{\nu} - \mathrm{id}): \gamma>0, \ \nu \in \mathcal{W}_2(\mathcal{X})\right\}$ is a pre-Hilbert space, the closure of which is the tangent space $\mathrm{Tan}(\mu) \subset \mathcal{L}^2(\mu)$ \citep{ambrosio2008}. Note that each $\mathrm{Tan}(\mu)$ consists, as a set, of square-integrable functions, so that different measures yield tangent spaces differing in their inner product rather than in their elements.

The projection onto the tangent space at $\mu$ is effected by the logarithm map $\nu \mapsto \log_{\mu}(\nu) \in \mathrm{Tan}(\mu)$, which records the direction and the magnitude of the transport required to move from $\mu$ to $\nu$, and which for absolutely continuous measures is given by
\begin{equation}\label{eq:log-map}
    \log_\mu(\nu)= F_\nu^{-1} \circ F_\mu-\mathrm{id} = T_{\mu}^{\nu} -\mathrm{id}.
\end{equation}
Perturbations are mapped back onto the space of measures by the exponential map $\exp_{\mu}(v) = (\mathrm{id}+v)_{\#}\mu$, which is a left inverse of the logarithm map, since $\exp_\mu(\log_\mu(\nu)) = (T_\mu^\nu)_\#\mu = \nu$.

 Let $\nu_1,\dots,\nu_n$ be absolutely continuous elements of $\mathcal{W}_2(\mathcal{X})$. Their 2-Wasserstein \textit{Fr\'echet} mean, or \textit{barycenter}, is a minimiser of the sum of squared distances,
$$
\hat\mu = \argmin_{\mu \in \mathcal{W}_2(\mathcal{X})}\sum_{i=1}^n W_2^2(\mu, \nu_i),
$$
and generalises the Euclidean mean to a setting in which simple averages are unavailable. Questions of existence and uniqueness of Fr\'echet means are in general delicate, but in the univariate Wasserstein space the barycenter exists, is unique, and is characterised through its quantile function \citep{Panaretos2016}:
\begin{equation}\label{eq:fpe}
F^{-1}_{\hat\mu} (\cdot) = \frac{1}{n}\sum_{i=1}^n F^{-1}_{\nu_i}(\cdot),
\end{equation}
as is immediate from \eqref{eq:w2-quantile}.

\subsection{Modelling and Registration} \label{sec:registration}
The temporal dynamics of the pandemic are captured by viewing the daily counts of a given state as a point process on a compact time window $[0,T]$, randomly distorted in time. Concretely, we posit a latent point process $\Pi$ with deterministic mean measure $\lambda$, together with i.i.d.\ random homeomorphisms $\varphi_1,\dots,\varphi_N$ of $[0,T]$, one per state, so that what we observe are the warped realisations $\widetilde\Pi_n = \varphi_{n\#}\Pi_n$ for $n=1,\dots,N$. The warps $\varphi_n$ then carry the temporal information we are after, and disentangling them from the realisations $\Pi_n$ is the content of the registration procedure of \citet{Panaretos2016}, which we now describe, following closely the authors' exposition.

Formally, we take $\Pi$ to be a point process on $[0,T]$, regarded as a random discrete measure with bounded second moment, $\mathbb{E}[(\int_{0}^{T}\mathrm{d}\Pi)^2] < \infty$, whose mean measure $\lambda(A) = \mathbb{E}[\Pi(A)]$ is defined on the Borel sets $\mathcal{B}$ of $[0,T]$. Recall that the push-forward of $\Pi_n$ through $\varphi_n$ acts as $\varphi_{n\#}\Pi_n(A) = \Pi_n(\varphi^{-1}_n(A))$, so that warping relocates mass in time without creating or destroying it. Introducing the random measures $\Lambda_n := \varphi_{n\#}\lambda$, that is $\Lambda_n(A) = \lambda(\varphi^{-1}_n(A))$, the mean measure of an observed process conditionally on its own time change is $\mathbb{E}[\widetilde{\Pi}_n \mid \varphi_n] = \Lambda_n$. Separating amplitude from phase variation thus amounts to constructing estimators $\{\hat{\varphi}_n\}$ and $\{\hat{\Pi}_n\}$ of the time changes and of the underlying realisations: the former capture the phase variation, the latter the amplitude variation.

Consistent separation requires two assumptions on the random time changes:
\begin{enumerate}
    \item \emph{Unbiasedness}: $\mathbb{E}[\varphi(x)]=x$ for every $x \in [0,T]$;
    \item \emph{Regularity}: $\varphi$ is monotonically increasing almost surely.
\end{enumerate}
The first fixes the average time change to be the identity, so that the ``objective'' time scale is preserved on average. Since our analyses are comparative and this objective scale is of no interest in itself, unbiasedness is of no practical consequence here. The second demands that $\varphi$ be a genuine warping of time, its failure amounting to a time reversal -- which is incompatible with most applications, ours included.

Under these assumptions the model admits a canonical representation, which in turn yields a straightforward estimation procedure. The key observation is that the structural mean measure $\lambda$ is the Fr\'echet mean of the random measure $\Lambda = \varphi_{\#}\lambda$ with respect to the Wasserstein metric $\argmin_{\mu\in \mathcal{W}_2(0,T)}\mathbb{E}\,W_2^2(\mu,\Lambda)$,

a natural estimator of which is given by the empirical Fr\'echet mean of $\{\Lambda_1,\dots,\Lambda_N\}$, namely $\argmin_{\mu}\sum_{n=1}^{N}W_2^2(\mu,\Lambda_n)$, which in one dimension is available in closed form by averaging quantile functions, cf.~\eqref{eq:fpe}.
The $\Lambda_n$ are of course not observable; what we observe are the point processes $\widetilde{\Pi}_n$. The procedure of \citet{Panaretos2016} therefore proceeds by substituting estimates of the $\Lambda_n$ obtained from the observed processes, in four steps.
\begin{enumerate}[label=(\arabic*),noitemsep]
    \item \emph{Estimate the random measures.} For every $n$, estimate $\Lambda_n$ by smoothing the observed process $\widetilde{\Pi}_n$, as described in Section~\ref{sec:data}, yielding $\hat{\Lambda}_n$.
    \item \emph{Estimate the structural mean.} Take $\hat{\lambda}$ to be the empirical Fr\'echet mean of $\hat{\Lambda}_1,\dots,\hat{\Lambda}_N$, available in closed form through \eqref{eq:fpe} as $F_{\hat\lambda}^{-1} = N^{-1}\sum_{n=1}^{N} F_{\hat\Lambda_n}^{-1}$.
    \item \emph{Estimate the time changes.} Take $\hat{\varphi}_n$ to be the optimal transport map from $\hat{\lambda}$ onto $\hat{\Lambda}_n$, that is $\hat{\varphi}_n = T_{\hat\lambda}^{\hat\Lambda_n} = F_{\hat\Lambda_n}^{-1}\circ F_{\hat\lambda}$, which is legitimate since an increasing map pushing $\lambda$ onto $\Lambda_n$ is necessarily the optimal one.
    \item \emph{Register.} Undo the estimated time changes, $\hat{\Pi}_n = (\hat{\varphi}^{-1}_{n})_{\#} \widetilde{\Pi}_n$ with $\hat{\varphi}^{-1}_n = F_{\hat\lambda}^{-1}\circ F_{\hat\Lambda_n}$, for $n = 1,\dots,N$.
\end{enumerate}
The maps $\hat\varphi_n$ and the registered processes $\hat\Pi_n$ are the estimated phase and amplitude components respectively. We refer the reader to \citet{Panaretos2016} for further details and for consistency results.

\subsection{Wasserstein Tangent-space PCA} \label{sec:WTPCA}

When data arises in the form of probability densities, their statistical analysis is complicated by the lack of a \textit{linear} structure on the space of measures, which renders the direct use of basic tools such as PCA inaccessible. Fortunately, the space of univariate measures, endowed with the geometric structure stemming from optimal transport theory, admits a \textit{linearisation} by means of a \textit{tangent space projection}.

By \eqref{eq:fpe}, projecting the measures $\nu_1,\dots,\nu_n$ onto the tangent space at their barycenter yields \emph{centred} objects, $\frac{1}{n}\sum_{i=1}^n\log_{\hat\mu}(\nu_i) = 0 \in \mathrm{Tan}(\hat\mu)$: the measure-valued data have been turned into zero-mean elements of a Hilbert space of square-integrable functions, in which an ordinary principal component analysis can be carried out.

In our setting, WTPCA reduces to the FPCA of the quantile functions. Indeed, for any absolutely continuous $\mu$ there is an isometry between $\mathcal{L}^2(\mu)$ and $\mathcal{L}^2(0,1)$ amounting to a change of measure \citep{bigot2017geodesic}: letting
\[
G: \mathcal{L}^2(\mu)\to \mathcal{L}^2(0,1) \qquad \text{be given by} \qquad (Gv)(u) = v(F^{-1}_\mu(u)),
\]
we have, for $v_1,v_2 \in \mathcal{L}^2(\mu)$,
\[
\langle v_1 , v_2 \rangle_{\mathcal{L}^2(\mu)} = \int_{\mathcal{X}} v_1(t)v_2(t)\,\mathrm{d}\mu(t) = \int_0^1 v_1(F^{-1}_\mu(u)) v_2(F^{-1}_\mu(u))\,\mathrm{d}u = \langle G v_1, G v_2 \rangle_{\mathcal{L}^2(0,1)},
\]
where the middle equality uses the fact that if $U$ is uniformly distributed on $(0,1)$ then $F^{-1}_\mu(U)$ has law $\mu$, applied to $h = v_1 v_2$. Applying $G$ to the log-maps \eqref{eq:log-map} moreover yields
\[
G \log_\mu(\nu) = F^{-1}_\nu \circ F_\mu \circ F^{-1}_\mu - \mathrm{id} \circ F^{-1}_\mu = F^{-1}_\nu - F^{-1}_\mu,
\]
so that the quantile-level data cloud $\{ F^{-1}_{\nu_i} - F^{-1}_{\hat\mu}\}_{i=1}^n \subset \mathcal{L}^2(0,1)$, with empirical covariance $\hat{C}$, is the image under $G$ of the tangent-level cloud $\{\log_{\hat\mu}(\nu_i) \}_{i=1}^n \subset \mathrm{Tan}(\hat\mu)$, with empirical covariance $\hat{S}$. The two covariance operators are therefore unitarily equivalent, $\hat{C} = G \hat{S} G^{-1}$, so they share their eigenvalues, and their eigenfunctions are related by $\psi_k = G \phi_k$, where $\{\psi_k\}$ and $\{\phi_k\}$ denote the eigenfunctions of $\hat{C}$ and of $\hat{S}$ respectively. At the quantile level, we thus obtain the usual Karhunen--Lo\`eve representation
\begin{equation}\label{eq:kl-quantile}
F^{-1}_{\nu_i} - F^{-1}_{\hat\mu} = \sum_k \xi_{ik}\, \psi_k, \qquad \xi_{ik} = \langle F^{-1}_{\nu_i} - F^{-1}_{\hat\mu},\, \psi_k\rangle_{\mathcal{L}^2(0,1)}.
\end{equation}

Suppose two measures differ in the $k$-th score alone, so that $F^{-1}_{\nu_j} = F^{-1}_{\nu_i} + \Delta\,\psi_k$ with $\Delta := \xi_{jk} - \xi_{ik}$. Then, for every level $u$, the quantile of $\nu_j$ at $u$ is that of $\nu_i$ shifted by $\Delta\,\psi_k(u)$: mass sitting at rank $u$ under $\nu_i$ is relocated by an amount $\Delta\,\psi_k(u)$ under $\nu_j$. The principal components are, in this sense, modes of deformation, or of transport, describing how a measure departs from the barycenter. The description is exact in the tangent space, and is carried back to the space of measures by the exponential map at $\hat\mu$.

Finally, note that the cloud $\{ F^{-1}_{\nu_i} - F^{-1}_{\hat\mu}\}_{i=1}^n$ is centred at zero by virtue of \eqref{eq:fpe}, whereas $\{F^{-1}_{\nu_i}\}_{i=1}^n$ is centred at the unknown $F^{-1}_{\hat\mu}$, which \eqref{eq:fpe} estimates. The ordinary FPCA of the quantile functions is therefore equivalent to WTPCA, which in practice reduces to three steps:
\begin{enumerate}[label=(\roman*),noitemsep]
    \item evaluate each quantile function $F_{\nu_i}^{-1}$ on a common equispaced grid $u_1 < \cdots < u_m \in (0,1)$, stacking the results as the rows of a matrix $\mathbf{Q} \in \mathbb{R}^{n \times m}$;
    \item centre $\mathbf{Q}$ column-wise, the column means estimating the quantile function of the Fr\'echet mean as in \eqref{eq:fpe};
    \item\label{item:svd} compute the truncated singular value decomposition of the centred matrix, $\mathbf{Q}_c = \mathbf{U}\mathbf{D}\mathbf{V}^\top$.

\end{enumerate}
The columns of $\mathbf{V}$, rescaled by $\sqrt{m}$, are the estimated eigenfunctions $\hat\psi_k$ evaluated on the grid; the entries of $\mathbf{U}\mathbf{D}/\sqrt{m}$ are the corresponding scores $\hat\xi_{ik}$; and $\hat\sigma^2_k = d_k^2/(nm)$, with $d_k$ the $k$-th singular value, estimates the $k$-th eigenvalue. The factors of $\sqrt{m}$ merely convert the Euclidean inner product on the grid into its $\mathcal{L}^2(0,1)$ counterpart, and cancel in the reconstruction \eqref{eq:kl-quantile}.

\paragraph{Consistency.} It remains to record that the procedure is statistically consistent, both under direct observation of the measures and in the practically relevant scenario of point process observations, which is the one arising in our application.

Assume a generative model in which each observed distribution is a push-forward of a fixed absolutely continuous reference measure $\mu$ on a compact interval $\mathcal{X}\subset\mathbb{R}$ with $F_\mu$ strictly increasing on the interior of $\mathcal{X}$,
\begin{equation*}
  \nu_i=(T_i)_\#\mu,\qquad i=1,\dots,n,
\end{equation*}
the $T_i$ being i.i.d.\ copies of a random \emph{increasing} homeomorphism $T:\mathcal{X}\to\mathcal{X}$ with $\mathbb{E}[T]=\id$ and $\mathbb{E}\|T-\id\|^2_{L^2(\mu)}<\infty$.

We assume the variability of $T-\id$ in $\mathrm{Tan}(\mu)$ to be of finite rank $R$: writing $\{(\sigma_j^2,\phi_j)\}_{j\ge1}$ for the spectral decomposition of the covariance operator $\mathbb{E}[(T-\id)\otimes(T-\id)]$, we take $\sigma_j^2=0$ for all $j>R$, with $\sigma_1^2>\cdots>\sigma_R^2>0$.

\begin{proposition}
\label{prop:consistency}
Let $\nu_1,\dots,\nu_n$ be i.i.d.\ probability measures on a compact $\mathcal{X}\subset\mathbb{R}$ obeying the generative model above, let $\hat\mu$ denote the empirical barycenter, and set $\hat T_i:=T^{\nu_i}_{\hat\mu}$. Then the following hold.
\begin{enumerate}
  \item $W_2^2(\hat\mu,\mu)=O_p(n^{-1})$.
  \item Let $P:L^2(\mathcal{X},\hat\mu)\to L^2(\mathcal{X},\mu)$ denote the parallel transport $Pv=v\circ F_{\hat\mu}^{-1}\circ F_\mu$, and let $\{(\hat\sigma_j^2,\hat\phi_j)\}$ be the spectral decomposition of $\frac1n\sum_{i=1}^n(\hat T_i-\id)\otimes(\hat T_i-\id)$, the signs of the $\hat\phi_j$ being chosen so that $\langle P\hat\phi_j,\phi_j\rangle\ge0$. Then
  \begin{equation*}
    \sup_{1\le j\le R}\bigl|\hat\sigma_j^2-\sigma_j^2\bigr|^2=O_p(n^{-1}),\qquad \sup_{1\le j\le R}\bigl\|P\hat\phi_j-\phi_j\bigr\|^2_{L^2(\mathcal{X},\mu)}=O_p(n^{-1}).
  \end{equation*}
  \item For every $i$, the rank-$R$ reconstruction is consistent:
  \begin{equation*}
    W_2\Bigl(\exp_\mu\Bigl(\sum_{j=1}^{R}\langle\hat T_i-\id,\hat\phi_j\rangle_{L^2(\mathcal{X},\hat\mu)}\,P\hat\phi_j\Bigr),\ \nu_i\Bigr)=o_p(1).
  \end{equation*}
\end{enumerate}
\end{proposition}

The estimated eigenfunctions and scores thus converge at the parametric $O_p(n^{-1/2})$ rate, and WTPCA admits a pointwise interpretation analogous to that of FPCA, cf.~Appendix~\ref{sec:FPCA}.

\begin{remark}[Imperfect observations]\label{rem:imperfect_obs}
While Proposition~\ref{prop:consistency} presumes direct observation of the measures, in practice one observes event times from point processes whose intensities are the unknown measures -- the setting of our application. \citet[][Theorem~1]{Panaretos2016} establishes that consistent estimation of the Fr\'echet mean remains possible in this setting under minimal conditions, so that the conclusions above carry over.
\end{remark}

\subsection{Vector-on-Vector Regression} \label{sec:multiv_reg}

Finally, we introduce the reader to multivariate (vector-on-vector) regression, which will be applied to \emph{the WTPC} scores mentioned above to analyse the restrictions' association with the evolution of COVID-19 infections, as shown in Section~\ref{sec:results}.

A vector-on-vector regression model is an extension of the usual linear regression model $\mathbf{y} = \mathbf{X \beta} + \boldsymbol{\epsilon}$ to the case of more than one dependent variable, which takes the following form: \[ \mathbf{Y} = \mathbf{X B} + \boldsymbol{\Xi}.\]
Here $\mathbf{Y} \in \mathbb{R}^{n \times p}$ is a matrix of $p$ observed response variables for each of the $n$ observational units (i.e.~each of the 50 states in our application), $\mathbf{X} \in \mathbb{R}^{n \times q}$ is a matrix of $q$ independent variables observed for each of the $n$ observational units, and $\boldsymbol{\Xi} \sim N_p(\mathbf{0}, \boldsymbol{\Sigma})$ is a matrix of unobserved randomly distributed disturbances whose rows for given $\mathbf{X}$ are uncorrelated.

The log-likelihood in terms of the matrix parameters $\mathbf{B}$ and $\mathbf{\Sigma}$ takes the form
\begin{equation*}
    l(\mathbf{B}, \mathbf{\Sigma}) = - \frac{n}{2} \text{log}|2 \pi \mathbf{\Sigma}| - \frac{1}{2} \text{tr} (\mathbf{Y} - \mathbf{XB}) \mathbf{\Sigma}^{-1} (\mathbf{Y} - \mathbf{XB})^T ),
\end{equation*}

which leads to the maximum likelihood estimator of $\mathbf{B}$ being:
\[ \hat{\mathbf{B}} = (\mathbf{X}^T \mathbf{X})^{-1} \mathbf{X}^T \mathbf{Y} .\]

The reader interested in the intermediate steps that lead to the above result can consult Chapter 6 of \citet{Mardia1979} for reference.

The estimator above captures the relationship between the dependent variables $\mathbf{Y}$ and the independent variables $\mathbf{X}$. Although the point estimates turn out to be the same as if $p$ separate linear models were fitted, proper statistical tests must account for the covarying columns~of~$\mathbf{Y}$~\citep{Maxwell2017}. This covariance must be considered when determining whether a predictor contributes to the response jointly. To address this, the multivariate analysis of variance with the Pillai test statistic is used \citep{Mardia1979, Maxwell2017}.

\section{Results}\label{sec:results}

We next present the result of our analysis, carried out on the fifty US states over their respective first-wave windows of the COVID-19 epidemics in 2020. It proceeds in three stages. We begin by registering the infection counts and examining what is left once the temporal dynamics have been factored out, which both illustrates the separation achieved by the procedure of Section~\ref{sec:registration} and shows that, for these data, the phase component is at least as informative as the amplitude one. We then apply WTPCA to the infection and stringency densities, obtaining for each state a small number of scores summarising their temporal dynamics, and likewise for their restrictions; these admit a transparent reading in terms of an overall time shift and of the flatness of the corresponding curve. Finally, we relate the two sets of scores through a vector-on-vector regression, from which our main finding emerges: it is the timing of the restrictions, rather than their overall extent, that associates with the shape of the infection curve.

As an initial exploratory step, we perform FPCA on the smoothed, registered, and log-transformed infection count curves. The mean and top two eigenfunctions are shown in Figure~\ref{fig:fpca}, while the corresponding score plot is shown in Appendix~\ref{sec:FPCA} (Figure~\ref{fig:fpca_scores}). The registration methodology of \citet{Panaretos2016} is consistent when the number of infection counts for each state is large. Since this is clearly the case here, it is reasonable to assume that registration has been successful. Under the Cox-process model, successful registration would leave intensity functions that differ only by multiplicative constants. On the log scale, the registered curves would therefore differ only by additive shifts. However, Figure~\ref{fig:fpca} (left) indicates that this is clearly not the case. While minor issues common to every real-data application should not be overemphasised, the relatively high percentage of variance explained (11~\%) and the clear interpretability of the second eigenfunction in Figure~\ref{fig:fpca} (right) provide rather strong applied evidence (cf.~the interpretation of FPCA in Appendix~\ref{sec:FPCA}) that the infection count processes are not rank one after registration.
Therefore, the Cox point process model adopted, for example, by \citet{gajardo2021cox,gajardo2023point} is arguably not the right model for infection counts. This can also be argued from the obvious lack of independence between individual infections \citep[cf.][]{li2023nonstationary, indriani2024lgcp, meyer2018self, reinhart2018review}.

\begin{figure}[t!]
  \centering
     \begin{tabular}{ccc}
           \includegraphics[width=0.31\textwidth]{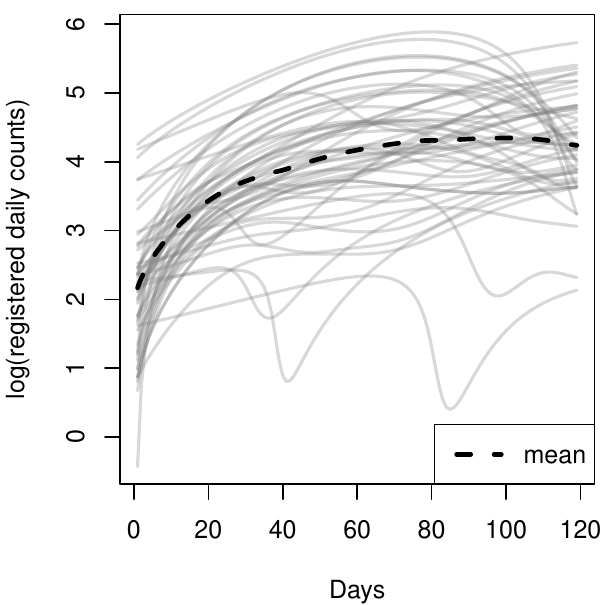}  &
           \includegraphics[width=0.31\textwidth]{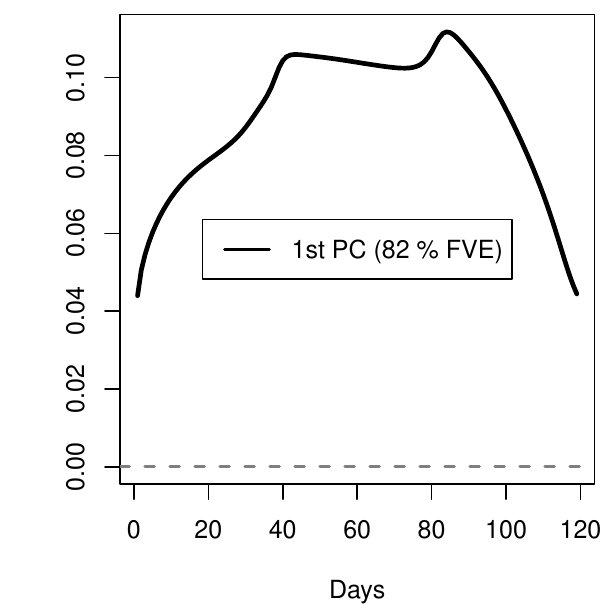} &
           \includegraphics[width=0.31\textwidth]{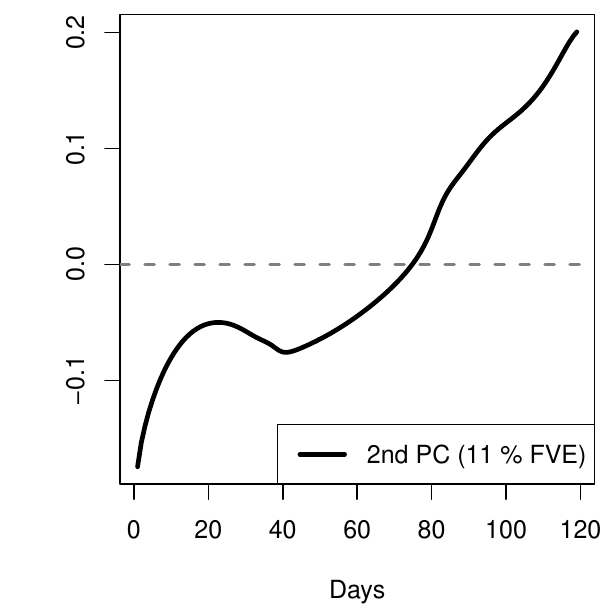}
    \end{tabular}
  \caption{FPCA of the registered infection daily counts (on log scale). \emph{
  From left to right, the plots show the log transform of registered counts, and the first and second PC eigenfunctions, which explain 82~\% and 11~\% of the variance respectively.}}
  \label{fig:fpca}
\end{figure}

\subsection{WTPCA of Infection Counts and the Stringency Index}

\begin{figure}[!htbp]
  \centering

  \begin{minipage}[t]{0.32\textwidth}
    \centering
    \includegraphics[width=\linewidth]{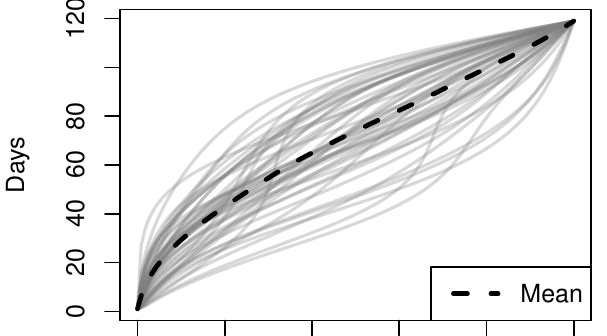}\\
    \includegraphics[width=\linewidth]{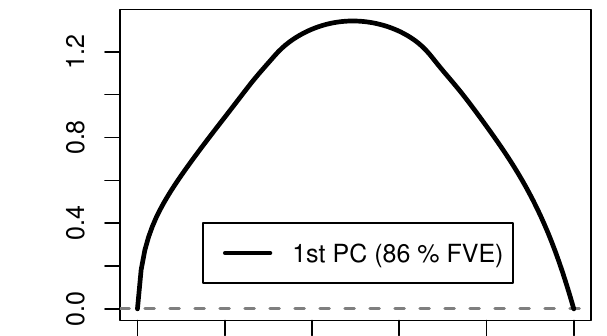}\\
    \includegraphics[width=\linewidth]{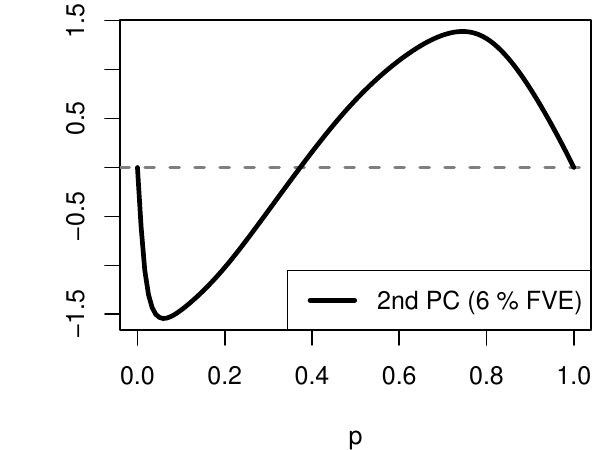}\\[1.5em]
    \includegraphics[width=\linewidth]{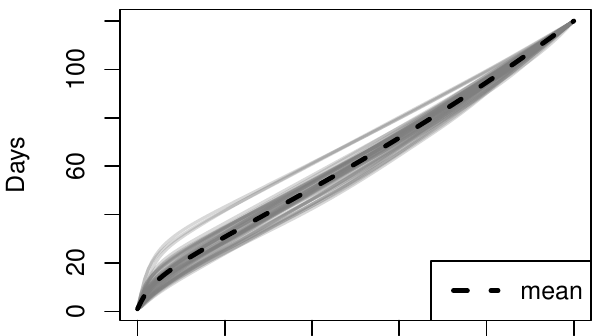}\\
    \includegraphics[width=\linewidth]{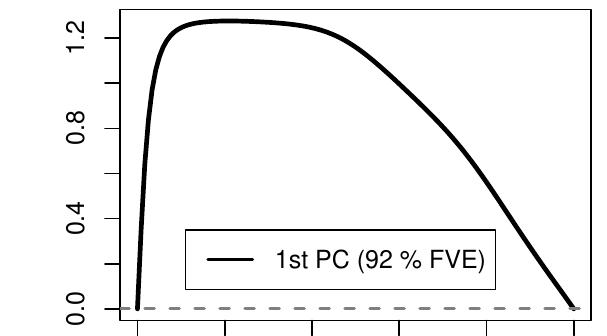}\\
    \includegraphics[width=\linewidth]{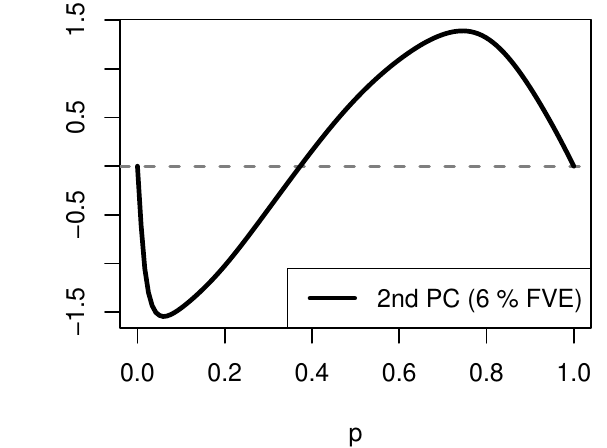}
  \end{minipage}
  \hfill
  \begin{minipage}[t]{0.6\textwidth}
    \centering
    \vspace{-3cm}
    \includegraphics[width=\linewidth]{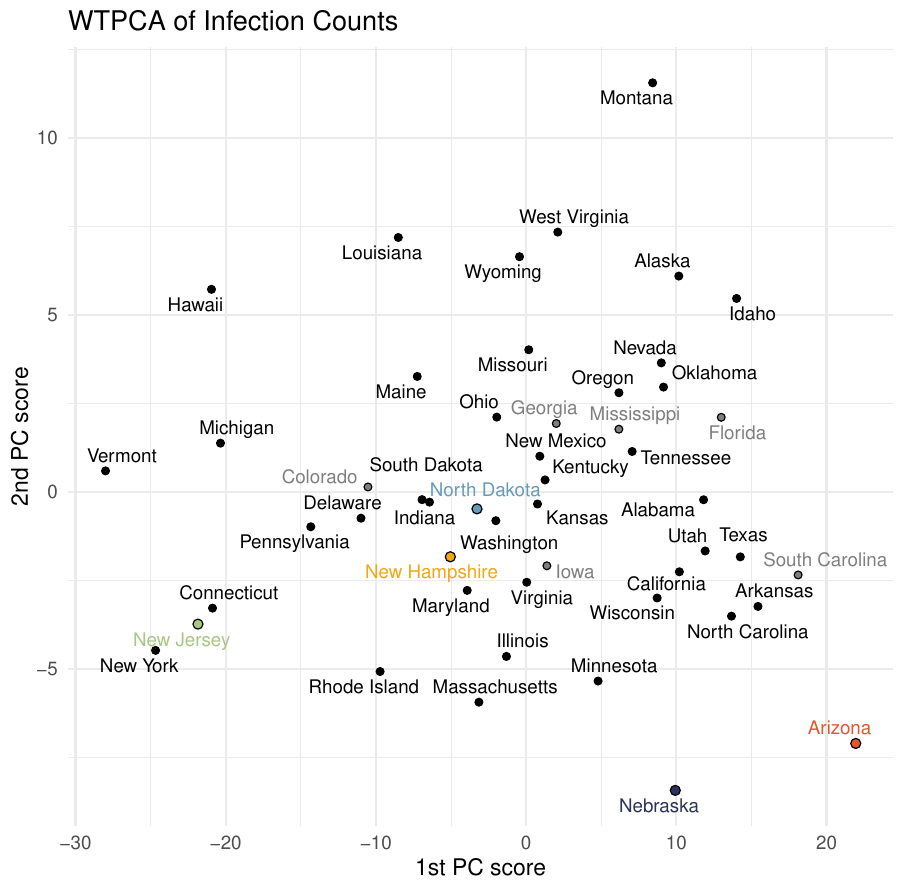}\\[1.5em]
    \includegraphics[width=\linewidth]{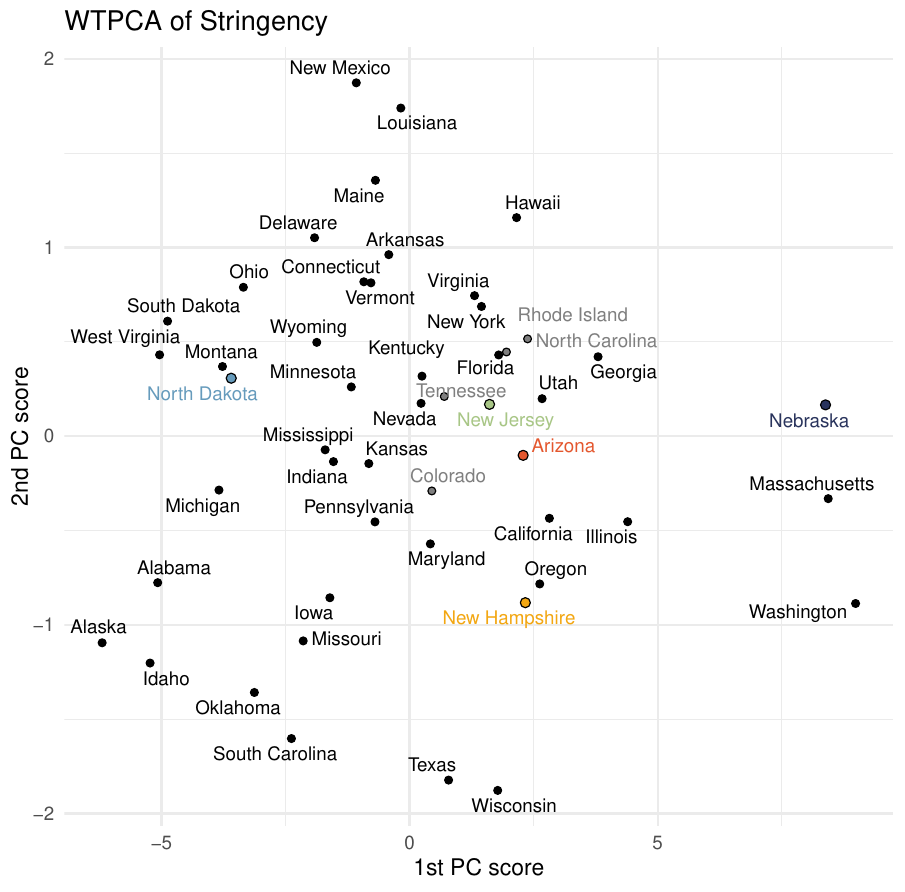}
  \end{minipage}

  \caption{Wasserstein PCA of the infection counts and the Stringency Index. 
  \emph{Left: quantiles and eigenfunctions for infection counts (top three panels) and for the Stringency Index (bottom three panels). Right: PC scores for infection counts (top) and Stringency Index (bottom). The colour convention follows that of Figure~\ref{fig:raw_data}, while some additional states are depicted in grey for the sake of better readability.}}
  \label{fig:WTPCA}
\end{figure}

Paving the way to the vector-on-vector regressions, we then performed PCA in the Wasserstein tangent-space as described in Section~\ref{sec:WTPCA} both for the infection counts and the stringency index. The results are shown in Figure~\ref{fig:WTPCA}. Each distribution $\nu_i$ is approximated by the rank-$R$ expansion
\begin{equation}\label{eq:WPCA_approx}
F_{\nu_i}^{-1}(u) \approx F_{\hat\mu}^{-1}(u) + \sum_{j=1}^{R} s_{i,j}\,\hat\phi_j(u), \qquad u \in (0,1),
\end{equation}
where $s_{i,j} = \langle F_{\nu_i}^{-1} - F_{\hat\mu}^{-1},\,\hat\phi_j\rangle_{L^2[0,1]}$ are the PC scores and $\hat\phi_j$ are the estimated eigenfunctions (right singular vectors of the centred quantile matrix, cf.~\ref{item:svd} above).

Interpreting these results requires some care, since the analysis is carried out on the quantile scale, where the roles of the two axes are interchanged relative to the densities one is accustomed to reading: a quantile function records the time by which a given fraction of the mass has accumulated, rather than the mass present at a given time. Consider the mean and the first two eigenfunctions of the Oxford stringency index, shown in the bottom three panels of the left column of Figure~\ref{fig:WTPCA}. By \eqref{eq:fpe}, the mean is the pointwise average of the state-level quantile functions, and is thus the quantile function of the Fr\'echet mean $\hat\mu$, while the eigenfunctions are the successively orthogonal directions of greatest variation of the data projected onto the tangent space at $\hat\mu$. It follows from \eqref{eq:kl-quantile} that a state whose score along the $k$-th component exceeds the average by $\delta$ has its $u$-th quantile displaced by $\delta\,\hat\psi_k(u)$, for every $u$. Figure~\ref{fig:WPCA_interpretation} (top-left) displays this for the first component, adding and subtracting a multiple of $\hat\psi_1$ from the mean quantile function; the bottom-left panel shows the same perturbation on the level of densities, obtained by mapping back through the exponential map at $\hat\mu$.

Since the 1st PC is all positive, we can assign it the following interpretation. Adding (resp.~subtracting) the 1st PC to any quantile function pushes the quantiles up (resp.~down) shifting the corresponding density right, resp.~left. The 1st PC thus captures the overall time shift. Similarly, the 2nd PC captures the flatness/spikedness of the pandemic evolution, as we next motivate. Since the 2nd PC is initially negative and later positive, it contrasts the lower quantiles to upper quantiles. And since the upper quantiles are naturally higher than lower quantiles, adding (resp.~subtracting) the 2nd PC either moves the lower and upper quantiles further apart, (resp.~closer together), as shown in Figure~\ref{fig:WPCA_interpretation} (top-right). As a result, the density becomes flatter (resp.~more spiky), as shown in Figure~\ref{fig:WPCA_interpretation}.

The interpretation of PCs belonging to the case counts is qualitatively similar. The 1st PC is again all positive, so moving along its direction in the quantile space pushes the corresponding density to the right (or left). The 2nd PC is again first negative and then positive, so moving along it in the quantile space influences the spikedness/flatness of the corresponding density.

\begin{figure}[!t]
  \centering
    \begin{tabular}{cc}
           \includegraphics[width=0.35\textwidth]{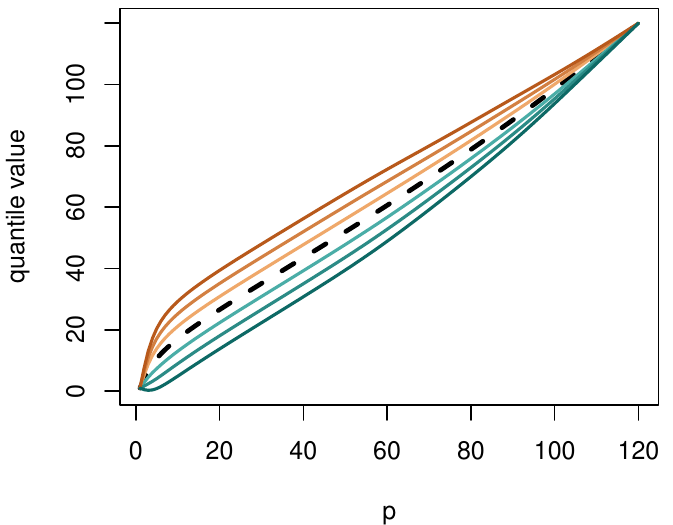}  &
           \includegraphics[width=0.35\textwidth]{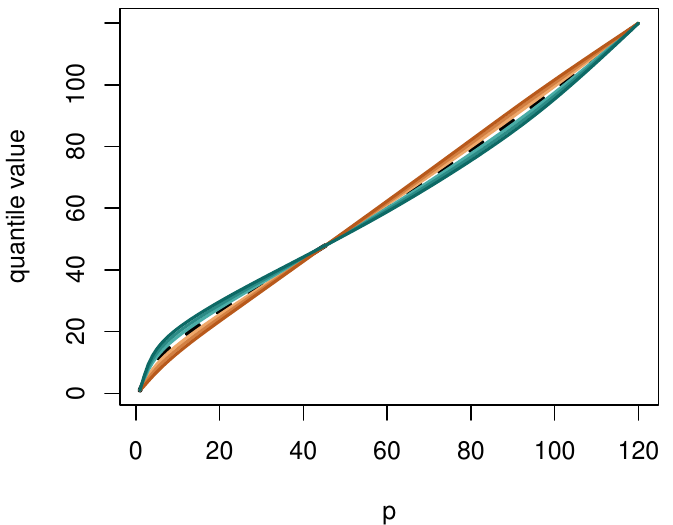} \\\\
           \includegraphics[width=0.35\textwidth]{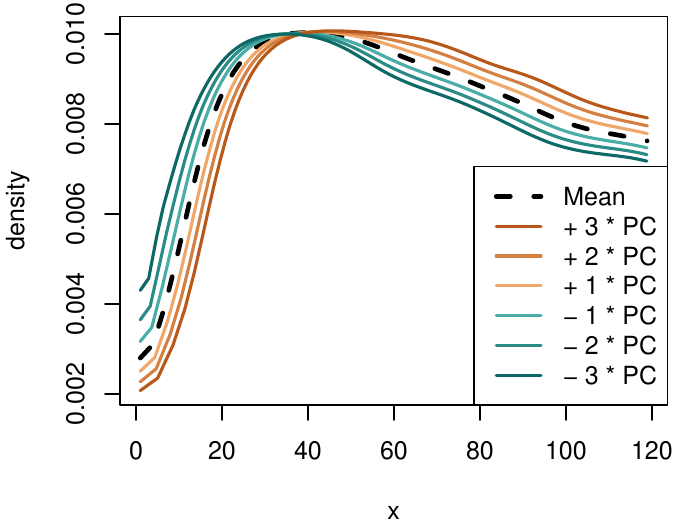}  &
           \includegraphics[width=0.35\textwidth]{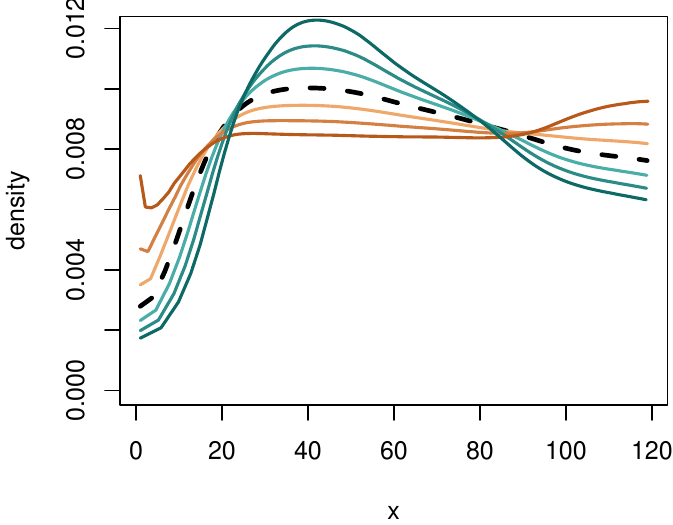}
  \end{tabular}
  \caption{Interpreting scores of Wasserstein PCA. \emph{The top row shows the mean quantile function of the stringency index (dashed) plus/minus constant times the 1st (left) and 2nd (right) eigenfunctions of the stringency index. The bottom row shows the same on the level of densities.}}
  \label{fig:WPCA_interpretation}
\end{figure}

Equipped with this interpretation, we can now easily assign meaning to the score plots, shown in the right column of Figure~\ref{fig:WTPCA}. Let us start with the cases. A positive 1st PC score suggests that the corresponding state had the underlying density of the infection counts pushed right compared to the mean, i.e.~that the 1st wave of the pandemic hit that state later than average. Secondly, a positive 2nd PC score is associated with a flatter than average corresponding density. For example, Arizona has an extremely high 1st PC score and an extremely low 2nd PC score, meaning that it was hit later and the infection counts spiked more than for an average state. Note that flatness of the infection counts is particularly desirable: several studies have shown that a flatter curve and, thus, a reduction in the peak of the outbreak is critical to reducing the burden on the healthcare-system capacity and, thus, ensuring care for the most fragile patients \citep{Miller2022, JohnsHopkins2020, CDC2007, Gavin2020}.

Similarly, for the Oxford Stringency Index, high 1st PC score means the restrictions were distributed later than average, while high 2nd PC score means these were distributed more uniformly over time. Thus, for example, Nebraska implemented the mitigation measures later than the bulk of the states, but was average in the temporal focus of the stringency. This can also be seen in the shape of Nebraska's stringency curve in Figure~\ref{fig:raw_data} (right). Secondly, states such as New Mexico or Louisiana had flatter stringency curves, while states such as South Carolina, Texas, or Wisconsin had more temporally focused mitigation measures. Finally, we should mention that the shape of the 1st eigenfunction of stringency suggests that the variability in stringency distribution between the states was highest for lower quantiles (say between 0.15 and 0.4). This aligns with the empirical observation that, after a while, most states aligned with nationwide restrictions, leading to a substantial reduction in differences in higher quantiles.

These scores serve as the inputs of the vector-on-vector regression model of Section~\ref{sec:multiv_reg}, alongside the two scalars set aside in Section~\ref{sec:data} -- the total infection count $c_n$ and the overall stringency budget $b_n$ of each state -- and the control variables, namely the log-GDP per capita and the log-population density. How many components to retain is largely settled in the case of the stringency index, whose first two already account for almost 99~\% of the tangent-space variance, but is less clear-cut for the infection counts, where the first two explain $>95\,\%$ and a third would add a further $<2.5\,\%$. We retain two components throughout, both for uniformity of presentation and because the conclusions of the regression are unaffected by the choice; we return to this point, and to the accuracy of the resulting rank-two approximation, in Section~\ref{sec:stability}.

\subsection{Regression Analysis and Interpretation}

It remains to relate the two sets of scores, which we carry out by vector-on-vector regression. For the $n$-th state we take a three-dimensional response, consisting of the total infection count $c_n$ together with the two case scores $\xi_{n1}, \xi_{n2}$, and regress it on the stringency budget $b_n$, the two stringency scores $\zeta_{n1}, \zeta_{n2}$, and the two controls $g_n$ (log-GDP per capita) and $d_n$ (log-population density). A multivariate model -- rather than three separate univariate ones -- is preferred, because of the clear relationship in the three responses, and this covariance has to be accounted for in assessing whether a predictor contributes to them jointly.

Table~\ref{tab:summary} shows the fitted model where the control variables are either included or excluded. The full model can be written in matrix form as
\begin{equation}\label{eq:model}
\E\begin{pmatrix}
    c_1 & \xi_{11} & \xi_{12} \\
    c_2 & \xi_{21} & \xi_{22} \\
    \vdots & \vdots & \vdots \\
    c_{50} & \xi_{50,1} & \xi_{50,2} \\
\end{pmatrix} =
\begin{pmatrix}
    1 & b_1 & \zeta_{11} & \zeta_{12} & g_1 & d_1 \\
    1 & b_2 & \zeta_{21} & \zeta_{22} & g_2 & d_2 \\
    \vdots & \vdots & \vdots & \vdots & \vdots & \vdots \\
    1 & b_{50} & \zeta_{50,1} & \zeta_{50,2} & g_{50} & d_{50} \\
\end{pmatrix}
\begin{pmatrix}
    \beta_{0,c} & \beta_{0,1} & \beta_{0,2} \\
    \beta_{1,c} & \beta_{1,1} & \beta_{1,2} \\
    \vdots & \vdots & \vdots \\
    \beta_{5,c} & \beta_{5,1} & \beta_{5,2} \\
\end{pmatrix},
\end{equation}
where the first index of a coefficient refers to the covariate, in the order in which the columns above are listed, and the second to the response: the subscript $c$ pertains to the total infection count, and the subscripts $1$ and $2$ to the first and second case scores respectively.

While Table~\ref{tab:summary} provides the fitted coefficients together with their standard errors and significance codes as fitted by \verb|lm()| in the statistical software package R \citep{R}, the significance is obtained from three univariate-response regression models. To assess joint significance, Table~\ref{tab:manova} performs multivariate analysis of variance using the Pillai test statistic.

\begin{table}[!t]
\centering
\caption{Regression outputs. \emph{Estimated coefficients together with their standard errors (in parentheses) and significance codes are given for the full Model and the Submodel (without the control covariates). \\Significance codes: ‘***’ $p<0.001$; ‘**’ $p<0.01$; ‘*’ $p<0.05$; ‘$\cdot$’ $p<0.1$.}}
\begin{tabular}{
  l
  S[table-format=-4.3]
  S[table-format=4.3]
  c
  S[table-format=-5.3]
  S[table-format=5.3]
  c
}
\toprule
& \multicolumn{3}{c}{\textbf{Submodel}} & \multicolumn{3}{c}{\textbf{Model}} \\
\cmidrule(lr){2-4} \cmidrule(lr){5-7}
\textbf{Covariates} & \multicolumn{1}{c}{Est.} & \multicolumn{1}{c}{S.E.} & \multicolumn{1}{c}{Sig.}
& \multicolumn{1}{c}{Est.} & \multicolumn{1}{c}{S.E.} & \multicolumn{1}{c}{Sig.} \\
\midrule
\multicolumn{7}{l}{\textbf{Response: Total Cases}} \\
(Intercept)         & 4188.8 & 7099.4 &    & -55650 & 34770 &     \\
Stringency\_PC1     & 29.30   & 17.74   &    & -18.08    & 17.78    &     \\
Stringency\_PC2     & 68.01   & 77.87   &    & 111.60    & 64.73    & $\cdot$   \\
Stringency\_Budget  & 0.39    & 0.87    &    & -0.92     & 0.77     &     \\
log\_GDP\_pc        &  &  &  & 5583  & 3148  & $\cdot$   \\
log\_pop\_dens     &  &  &  & 1967   & 449.3   & *** \\

\addlinespace
\multicolumn{7}{l}{\textbf{Response: PC1}} \\
(Intercept)         & 418.19  & 181.88 & *   & 4117.3 & 871.9 & *** \\
Stringency\_PC1     & -0.20   & 0.45   &     & 1.05    & 0.45   & *   \\
Stringency\_PC2     & -3.13   & 1.99   &     & -4.07   & 1.62   & *   \\
Stringency\_Budget  & -0.05   & 0.02   & *   & -0.03   & 0.02   &     \\
log\_GDP\_pc        &  &  &  & -341.11 & 78.93  & *** \\
log\_pop\_dens     &  &  &  & -27.81  & 11.27  & *   \\

\addlinespace
\multicolumn{7}{l}{\textbf{Response: PC2}} \\
(Intercept)         & -33.22  & 59.74  &     & 335.95  & 327.8 &     \\
Stringency\_PC1     & -0.64   & 0.15   & *** & -0.35   & 0.17   & *   \\
Stringency\_PC2     & 0.41    & 0.66   &     & 0.14    & 0.61   &     \\
Stringency\_Budget  & 0.004   & 0.007  &     & 0.01    & 0.007  &     \\
log\_GDP\_pc        &  &  &  & -34.44  & 29.68  &     \\
log\_pop\_dens     &  &  &  & -12.03  & 4.24   & **  \\

\bottomrule
\end{tabular}
\label{tab:summary}
\end{table}

We begin by interpreting the (simpler) submodel. The MANOVA table (Table~\ref{tab:manova}) reveals the 1st PC of stringency to be the only significant coefficient. From the summary table (Table~\ref{tab:summary}), it is apparent that this significance is due to the relationship between the 1st PC of stringency and the 2nd PC of the infection curves. The regression coefficient is negative, signalling that earlier stringency is associated with flatter infection counts.

When the important control variables are added, the situation changes. The effect reported in the previous paragraph is still present, but mitigated. The control variables alone are highly significant, most notably higher population density is associated with higher total cases, and higher GDP is associated with an earlier increase in infection counts, consistent with previous studies \citep{Ferreira2021, Gong2022}. More importantly, in the full model, both the control variables and the two PCs of stringency are significant. Apart from the control variables, a notable association appears between the 1st PC of cases and 2nd PC of stringency: a flatter stringency profile is associated with an earlier temporal distribution of cases. Secondly, the 1st PC of stringency is now also associated with the 1st and the 2nd PC of cases in the full Model: later stringency is associated with more spiked case curves. Put differently, earlier stringency goes together with flatter infection counts. the total stringency budget is not significant: an overall increase in stringency is not found to be associated either with the total infection count or with the temporal characteristics of the infection curve.

Finally, the results in Tables~\ref{tab:summary} and~\ref{tab:manova} suggest that no aspect of stringency is significantly associated with total infection counts and that the overall stringency budget is not significantly associated with any of the infection outcomes considered. We assess these conclusions jointly using a likelihood-ratio test comparing the full model with the restricted model obtained by setting the following coefficients in the mean structure~\eqref{eq:model} to zero:
\[
H_0: \beta_{3,c}=\beta_{2,c}=\beta_{1,c}=\beta_{1,1}=\beta_{1,2}=0.
\]
The resulting p-value is approximately 0.2, so the null hypothesis is not rejected. Thus, the joint test provides no statistically significant evidence against the two aforementioned conclusions.

\begin{table}[!t]
\centering
\caption{Multivariate analysis of variance (MANOVA). \emph{The table reports the Pillai test statistic and p-values of the MANOVA tests performed on the coefficients of the Submodel and the full Model to test the regressors' joint contribution to the response variables. Significance codes: ‘***’ $<$ 0.001; ‘**’ $<$ 0.01; ‘*’ $<$ 0.05}.}
\begin{tabular}{lcccc}
\toprule
~ & \multicolumn{2}{c}{\textbf{Submodel}} & \multicolumn{2}{c}{\textbf{Model}} \\
\textbf{Covariate} & \textbf{Statistic} & \textbf{P-value} & \textbf{Statistic} & \textbf{P-value} \\
\midrule
Stringency budget & 0.115 & 0.142 \textcolor{white}{**} & 0.095 & 0.239 \textcolor{white}{*} \\
Stringency PC 1 & 0.295 & 0.001 ** & 0.219 & 0.015 * \\
Stringency PC 2 & 0.084 & 0.272 \textcolor{white}{**} & 0.184 & 0.034 * \\
log\_GDP\_pc & ~ & ~ & 0.353 & $< 10^{-3}$ *** \\
log\_pop\_dens & ~ & ~ & 0.394 & $< 10^{-3}$ *** \\
\bottomrule
\end{tabular}
\label{tab:manova}
\end{table}

To summarise, the fitted models associate the \emph{timing} of the restrictions, rather than their overall extent, with the shape of the infection curves: states whose stringency mass was placed earlier tend to display flatter case curves, possibly reflecting a slower temporal progression of infections, whereas the total stringency budget shows no discernible association with the total number of cases. We stress that these are associations and not effects, for the reasons set out in Section~\ref{sec:stability}, and that the contrast between the two is what the analysis speaks to, rather than the merits of restrictions as such.

\subsection{Stability Analysis and Potential Deficiencies}\label{sec:stability}
The data analysis in the previous section naturally depends on several judgement calls we needed to make. These include the rules leading to the time domain choices for every state (cf.~Figure~\ref{fig:time_windows}), the smoothing applied to the case counts as well as to the stringency, the transformations of the control variables, and the specific choice of the control variables (e.g.~which election should be used for the political control). Two further significant choices were to retain exactly two WTPCs for both the infection counts and the stringency, and to exclude the District of Columbia from the data set as a (rather obvious) outlier and leverage point.

The accompanying repository contains scripts that can be used to reproduce the analysis and in which, moreover, the judgement calls can be easily tweaked. While we believe the choices made in this manuscript to be sensible ones, it is comforting that the results are stable with respect to them. Indeed, we found that the most important qualitative conclusions -- that earlier stringency is associated with flatter case curves, and that a larger stringency budget shows no significant association with lower case counts -- are quite robust and can be reached regardless of the specific choices, as long as the core first wave time window (from mid-March to mid-June) is included.

This is reassuring, since the point process, phase-variability-oriented viewpoint we adopt depends on capturing a realisation of the same underlying process for each observation. While the perfect choice of the temporal window is difficult to make for every state, the stability just described encourages us not to consider this a big problem.

An important aspect of our approach is that it requires the standardised COVID-19 infection counts of the different US states to be i.i.d. This assumption is presumably not satisfied here, but the same criticism can be raised in most applications, and the model diagnostics in Appendix~\ref{sec:regression-diagnostics} did not reveal any systematic or spatial patterns.

Furthermore, one has to be careful when interpreting the outputs of the regression modelling. No causal conclusions should be drawn from the model. Ours is an observational study, in which responses and covariates stand in a clear feedback loop; this amounts to simultaneous causality bias, which is notoriously difficult to correct for \citep{Roberts2013,Ketokivi2017}. Finally, the recorded case counts are an imperfect measurement of the underlying incidence: testing capacity was severely limited in the early weeks of the pandemic and expanded at a pace that differed across states, so that the early portion of each curve is undercounted to an extent that is itself related to when a~state was hit. All three issues call for careful judgement when interpreting the fitted regression model, and are the reason why we refrain from a quantitative interpretation and from any causal reading of the associations reported above.

\begin{figure}[!t]
  \centering
    \begin{tabular}{cc}
           \includegraphics[width=0.35\textwidth]{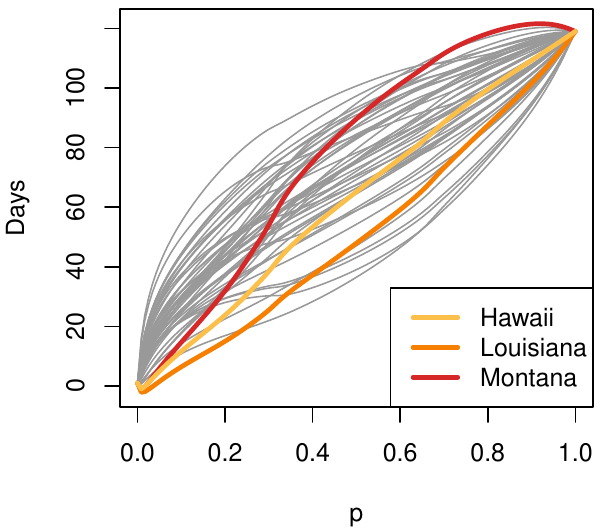}  &
           \includegraphics[width=0.35\textwidth]{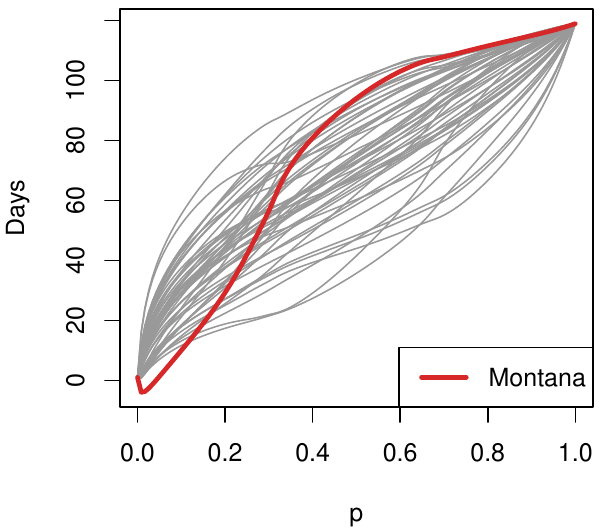}
  \end{tabular}
  \caption{WTPCA quantile projections of infection count curves. \emph{Left: rank two. Right: rank three.}}
  \label{fig:proj_quantiles}
\end{figure}

Further, WTPCA is naturally not the only possible approach to handle distributional data. While many alternatives exist \citep{carol2020,gajardo2023point,chen2023wasserstein,petersen2016functional,petersen2022modeling,pegoraro2022projected}, WTPCA lends itself to a simple interpretation that can be readily incorporated into the standard regression workflow. We thus believe WTPCA can serve as a conceptually close alternative to the common practice in regression analyses in which inherently distributional covariates are reduced to a few scalar summaries to be entered into the model.
Among these alternatives, one line of work \citep{bigot2017geodesic,cazelles2018geodesic,campbell2025efficient} replaces the linearisation around a reference measure by Geodesic PCA, which seeks principal paths along Wasserstein geodesics and thus respects the original geometry, at the price of a nonconvex optimisation problem even in the univariate case. This work also makes clear that WTPCA involves a double approximation: besides the finite-dimensional truncation, geodesic distances are replaced by linearised ones in the tangent space, the two agreeing only under a concentration condition that is (difficult to verify, but) interpretable as data concentration.

Figure~\ref{fig:proj_quantiles} shows the rank-2 and rank-3 projections of the data on the quantile level. We see that the rank-3 projections are more faithful at the edge of the domain, which implies that keeping three WTPCA scores instead of just two for the infection count curves might have been a better choice here. Moreover, even when rank three is chosen, Montana still remains problematic. However, this allows us to circle back to the stability analysis: we chose rank two for the simplicity of presentation and decided to keep Montana in the data set, and neither choice affects the conclusions drawn from our regression model. Another way of alleviating these approximation issues would be to adjust the studied state-specific time periods, which would again not alter the conclusions.

\section{Summary}

Despite the limitations discussed above, the main qualitative conclusion remains stable. Firstly, earlier restrictions are associated with flatter infection curves. Secondly, no aspect of stringency is significantly associated with total infection counts, and the overall stringency budget is not significantly associated even with the temporal characteristics of the infection curves.

On the other hand, a great strength of the proposed approach is its reliance on only mild modelling assumptions. Viewing the infection counts as point processes together with the other assumptions of \citet{Panaretos2016} is fairly uncontroversial, and the (degree of) smoothing chosen is well motivated in this application. Subsequently,
we provide a self-contained and practically implementable formulation of Wasserstein tangent-space PCA for distribution-valued covariates and responses. We also show how its scores, together with the corresponding total masses, can be incorporated into standard multivariate regression while retaining a direct transport-based interpretation.

Finally, this case study demonstrates how distributions can be incorporated in typical regression analyses by means of Wasserstein tangent-space PCA, which serves a similar regularisation purpose as functional PCA does in functional regression modelling. At the same time, it offers a~conceptually close alternative -- equally easy to estimate and explain -- to the summary-statistic approach to covariates that are naturally distribution-valued. Such covariates appear quite commonly in applied work; examples include age, income, household size, commuting time, test scores, pollution exposure, and other variables.

\phantomsection
\section*{Code Availability}\label{sec:code-availability}

Code to reproduce all analyses and figures as well as various stability checks is available at \url{https://github.com/FrancescoTripoli/Amplitude-phase-covid19}.

\setlength{\bibsep}{0pt plus 0.3ex}
\renewcommand{\baselinestretch}{1}
\bibliographystyle{rss}
\bibliography{biblio}

\clearpage
\appendix
\counterwithin{figure}{section}
\counterwithin{table}{section}
\renewcommand{\thefigure}{\thesection.\arabic{figure}}
\renewcommand{\thetable}{\thesection.\arabic{table}}
\setcounter{figure}{0}
\setcounter{table}{0}

\section{Functional PCA} \label{sec:FPCA}

PCA provides the most prominent data-driven dimension reduction technique \citep{Jolliffe2004},  and is often referred to as the workhorse of data analysis. In infinite dimensional contexts, its functional generalisation enjoys an even more prominent status \citep{happ2018multivariate}.
In this subsection we briefly review how to carry out this dimension-reduction technique in the general setting of functional data valued in  $\mathcal{L}^2[0,1]$.

Let $X$ be a random element in $\mathcal{L}^2[0,1]$ with mean $\mu$ and covariance operator $C$, where $C = \mathbb{E}[(X-\mu)\otimes (X-\mu)]$, and the tensor product generalises the euclidean outer product. The spectral theorem \citep{mercer1909xvi} entails that we may express the covariance in terms of its spectral eigendecomposition:
$$
C = \sum_{j=1}^\infty \lambda_j\, e_j \otimes e_j,
$$
where $\{ \lambda_j \}_{j=1}^\infty $ is a non-increasing, non-negative sequence of eigenvalues and $\{ e_j \}$ is an orthonormal basis in $\mathcal{L}^2[0,1]$ of eigenvectors. Naturally, this gives the following expansion for the random element $X$ itself in terms of the same basis:
\begin{equation}\label{eq:KL}
    X -\mu = \sum_{j=1}^\infty \langle X-\mu, e_j \rangle e_j.
\end{equation}
It can be easily shown that $\{ \langle X, e_j \rangle \}_{j=1}^\infty$ are uncorrelated random variables with mean zero and variances $\{ \lambda_j \}_{j=1}^\infty $.

Equation~\eqref{eq:KL} is the core of (functional) PCA. It allows us to decompose the stochastic fluctuations of the process $X$ around its mean into random but scalar constituents, i.e.~the scores \{$\langle X, e_j \rangle$\}, and a functional but deterministic part, i.e.~the eigenfunctions $\{e_j\}$. Since $\{e_j\}$ are orthonormal, they can be understood as different directions in $\mathcal{L}^2[0,1]$ that capture the separate modes of variation. And, since the scores \{$\langle X, e_j \rangle$\} are uncorrelated, these modes of variation contribute to the overall variability of the process $X$ in a linearly independent manner. Finally, since the variance of the scores is non-increasing, truncating the sum on the right-hand side of~\eqref{eq:KL} at some finite $r \in \N$ leads to the best $r$-dimensional approximation of the centred process $X - \mu$ \citep[][Thm.~7.2.8]{hsing2015}.

Estimation of the functional principal components is straightforward. Based on a random sample of curves $X_1,\ldots,X_N$, the mean and covariance are naturally estimated by their empirical counterparts
\[
\widehat{\mu} = \frac{1}{N} \sum_{n=1}^N X_n \;, \qquad \widehat{C} = \frac{1}{N} \sum_{n=1}^N (X_n - \widehat{\mu}) \otimes (X_n - \widehat{\mu}).
\]
In turn, taking the spectral eigendecomposition of $\widehat{C}$ we directly obtain estimates of the scores and eigenfunctions \cite{bosq2000linear}. By selecting a truncation parameter, e.g.~at rank $r$, the functional data $X_i$ can be represented as $r$-dimensional random vectors, with components given by the corresponding functional principal component scores $(\langle X_i, e_1 \rangle , \dots,\langle X_i, e_r \rangle )$.

The equalities and claims above are naturally understood in the $\mathcal{L}^2$-sense. However, when $X$ is mean-square continuous, they also hold in the uniform sense \citep{loeve1948functions,
karhunen1946spektraltheorie}. This is important for interpretation: we can understand eq.~\eqref{eq:KL} point-wise as $X(t) = \mu(t) + \sum_{j=1}^\infty \langle X, e_j \rangle e_j(t)$ and assign interpretation to the modes of variation $\{e_j\}$ in a point-wise manner as well. This in turn allows one to assign a domain-specific meaning to the eigenfunctions, revealing patterns in random fluctuation of the data.

\begin{figure}[b!]
  \centering
  \includegraphics[width=0.9\linewidth]{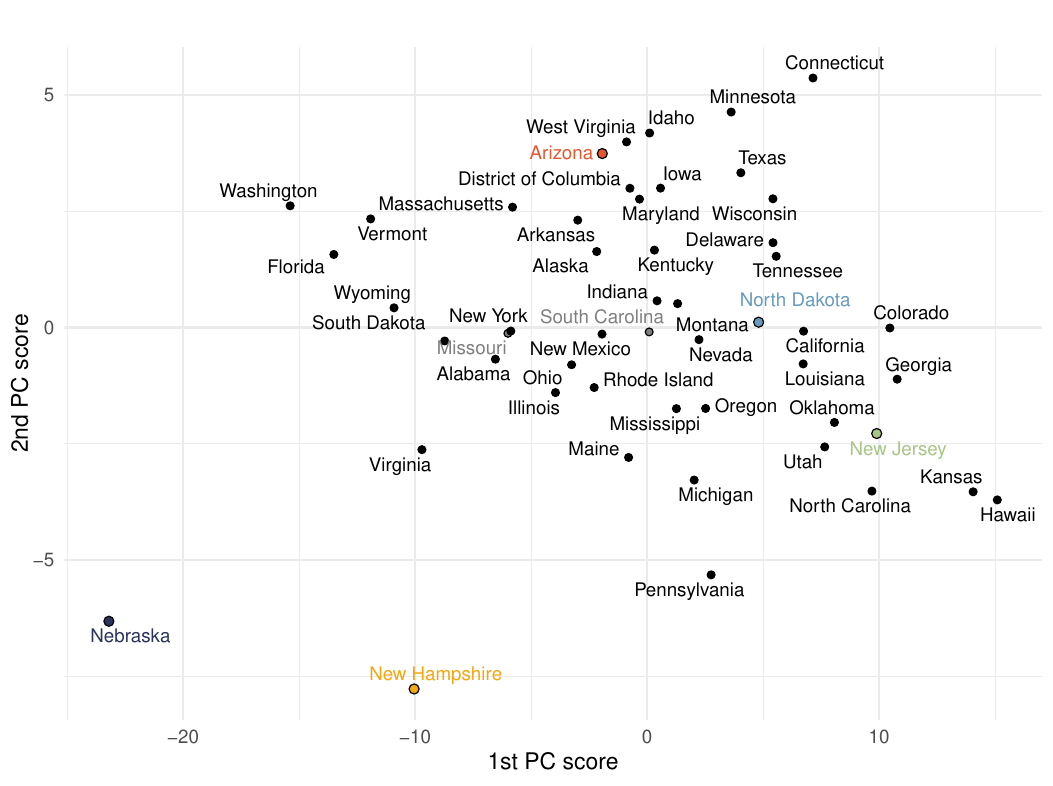}
  \caption{Score plot for the FPCA of log-registered daily infection counts. The corresponding eigenfunctions are shown in Figure~\ref{fig:fpca}.}
  \label{fig:fpca_scores}
\end{figure}

For example, Figure~\ref{fig:fpca} displays the eigenfunctions of the registered log-transformed infection count curves, leading to the two-dimensional representation of these through their PC scores in Figure~\ref{fig:fpca_scores}. The eigenfunctions show that the infection counts exhibit two main patterns of variation. In particular, the 1st eigenfunction is all positive capturing the overall level: the individual states reported infection numbers that either exceed (for corresponding positive scores) or fell short of the average, uniformly over the domain. Moreover, the upward trend indicates an increasing deviation of the daily counts from the average in the first 40 days, after which the variability remained roughly constant, and then plummeted as of day 90 of the pandemic. Consequently, states with higher scores for the 1st PC (those positioned far to the right in the score plot displayed in Figure~\ref{fig:fpca_scores}) had generally higher registered daily infection counts during the first wave. On the other hand, the 2nd eigenfunction captures a different mode of variation in infection numbers, contrasting the initial part of the studied period against its end. States with the 1st PC score near zero and a high 2nd PC score are characterised by initially reporting infection numbers below the average, catching up to the mean around about day 75, and continuing to exceed the mean daily infections for the rest of the period.

Based on the aforementioned interpretations associated with the principal components, keeping two components here allows us to discern four overarching qualitative patterns typically observed among states. Those situated in the first quadrant of the score plot Figure~\ref{fig:fpca_scores} -- exemplified by states like Texas or Connecticut, for which both scores are positive -- had relatively high infection counts overall, with a larger share occurring in the second half of the first-wave window. By contrast, states like Nebraska or New Hampshire, marked by negative scores for both components, were affected slightly more at the beginning but were less impacted compared to other states. Thirdly, intermediate scenarios exist, featuring states that experienced initial surge in infection counts and were impacted more overall (e.g.~Hawaii or New Jersey). And, finally, certain states displayed sharp growth initially, but later stabilised leading to overall below-average infection curve (e.g.~Washington or Vermont).
\section{Regression Diagnostics}\label{sec:regression-diagnostics}

For the full model summarised in Table~\ref{tab:summary}, we present residual and influence diagnostics here.

Firstly, Figure~\ref{fig:resid_plots} shows the standardised residuals of the three responses plotted against the fitted values, and their quantiles against the quantiles of the standard normal distribution. No serious heteroscedasticity or distributional issues are observed.

\begin{figure}[t!]
  \centering
  \includegraphics[width=0.9\linewidth]{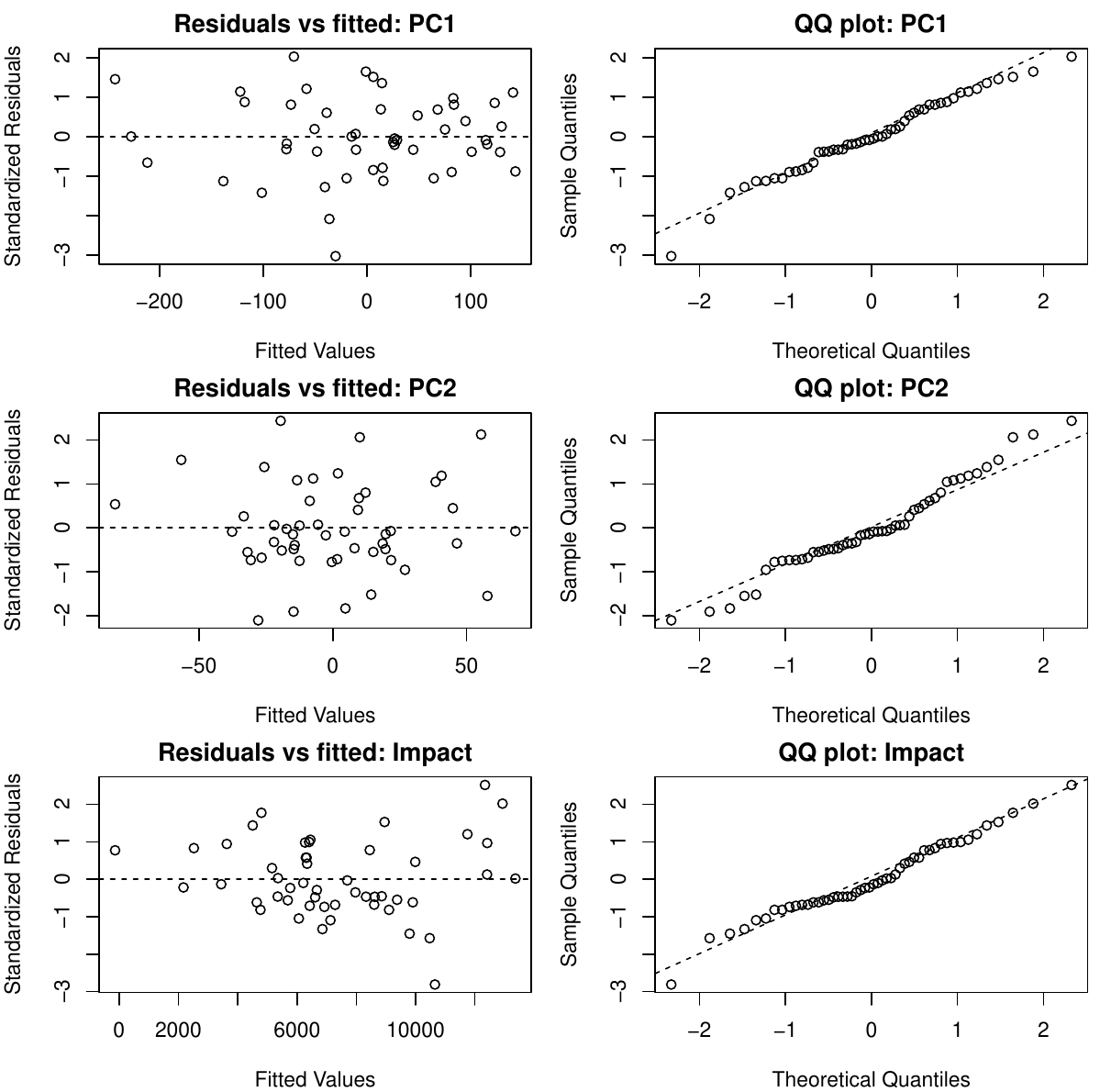}
  \caption{Standardised residuals vs.~fitted values (left) and QQ plots of the standardised residuals (right) for the three individual responses forming the response vector (the 1st response PC at the top row, the 2nd response PC in the middle row, and the total cases in the bottom row).}
  \label{fig:resid_plots}
\end{figure}

Secondly, in order to inspect the vector residuals jointly, we utilise the residual Mahalanobis distance and the multivariate Cook's distance \citep{seber2003}. Let $e_n \in \R^3$ be the residual vector corresponding to the $n$-th state, $n=1,\ldots,50$. The squared Mahalanobis distance $D_n^2$ is defined by
\[
D_n^2 = e_n^\top S_e^{-1} e_n,
\]
where $S_e$ is the empirical covariance matrix of the residuals. As such, the Mahalanobis distance measures how unusual an observation’s multivariate residual vector is. This is shown in Figure~\ref{fig:mahalanobis_cook}, left. In addition, the multivariate Cook's distance is given by the inner product
\[
C_n = \frac{1}{pq}\left\langle \widehat B-\widehat B_{-n}, \left(
S_e^{-1}
\otimes
X^{\top}X
\right) (\widehat B-\widehat B_{-n}) \right\rangle,
\]
where $S_e$ is as for the Mahalanobis distance above, $B_{-n}$ denotes the coefficient matrix estimate using the data set without the $n$-th observation, the $\otimes$ symbol denotes the outer product, and $pq$ denotes the total number of coefficients. Thus, similarly to the univariate Cook's distance, $C_n$ summarises the joint change in all response-specific regression coefficients caused by deleting a single observation $n$, with changes weighted according to the design geometry and the estimated covariance among the responses. Large values indicate observations with substantial influence on the fitted multivariate regression. For our model, the multivariate Cook's distances are displayed in Figure~\ref{fig:mahalanobis_cook}, right, with the common rule-of-thumb cutoff value $4/(nq)$.

\begin{figure}[t!]
  \centering
  \begin{tabular}{cc}
\includegraphics[width=0.45\linewidth]{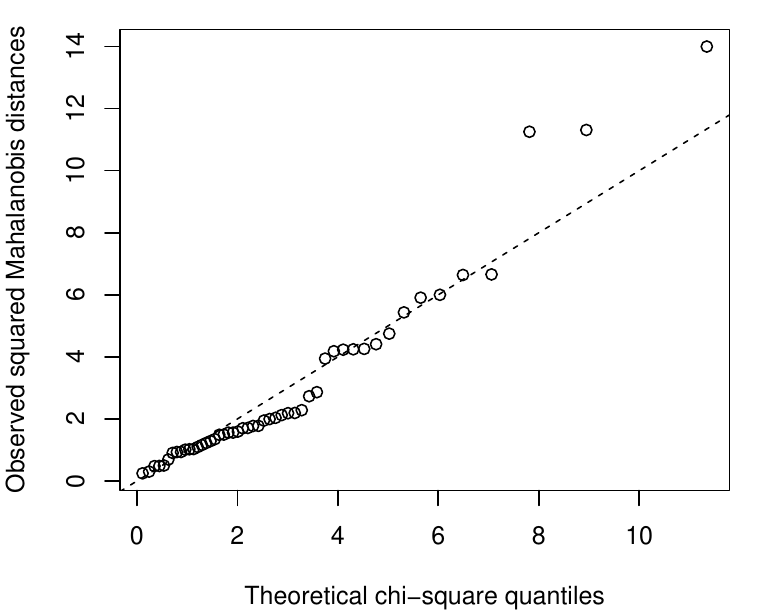}&
\includegraphics[width=0.45\linewidth]{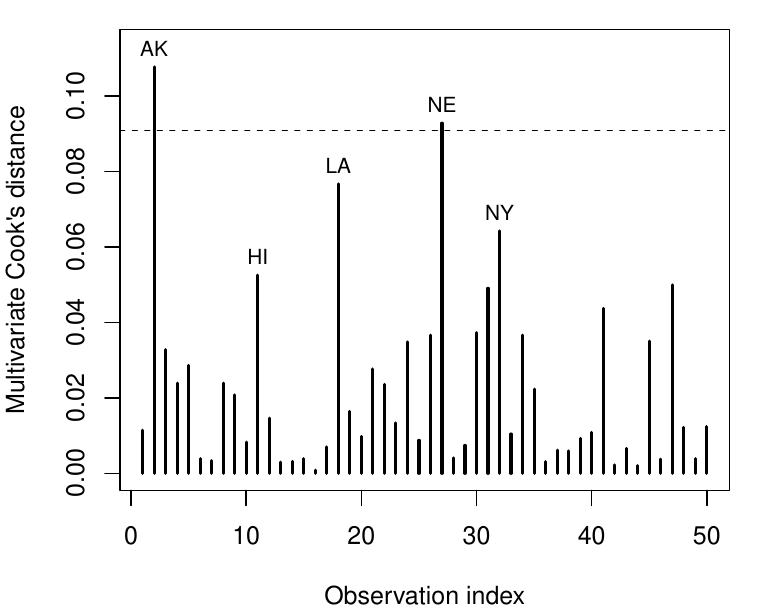}
  \end{tabular}
  \caption{The residual Mahalanobis distance (left) and the multivariate Cook's distance (right). The dashed vertical line for the Cook's distance is given by $4/(nq)$.}
  \label{fig:mahalanobis_cook}
\end{figure}

Thirdly, we examine potential residual dependencies. Figure~\ref{fig:US_map} shows a map of the US, where each state is coloured by its residual Mahalanobis distance. Furthermore, Figure~\ref{fig:boxplots} displays boxplots of the distances grouped by the four US Census regions. Again, we see no strong evidence for spatial or regional dependencies in the residuals.

\begin{figure}[]
  \centering
  \includegraphics[width=0.9\linewidth]{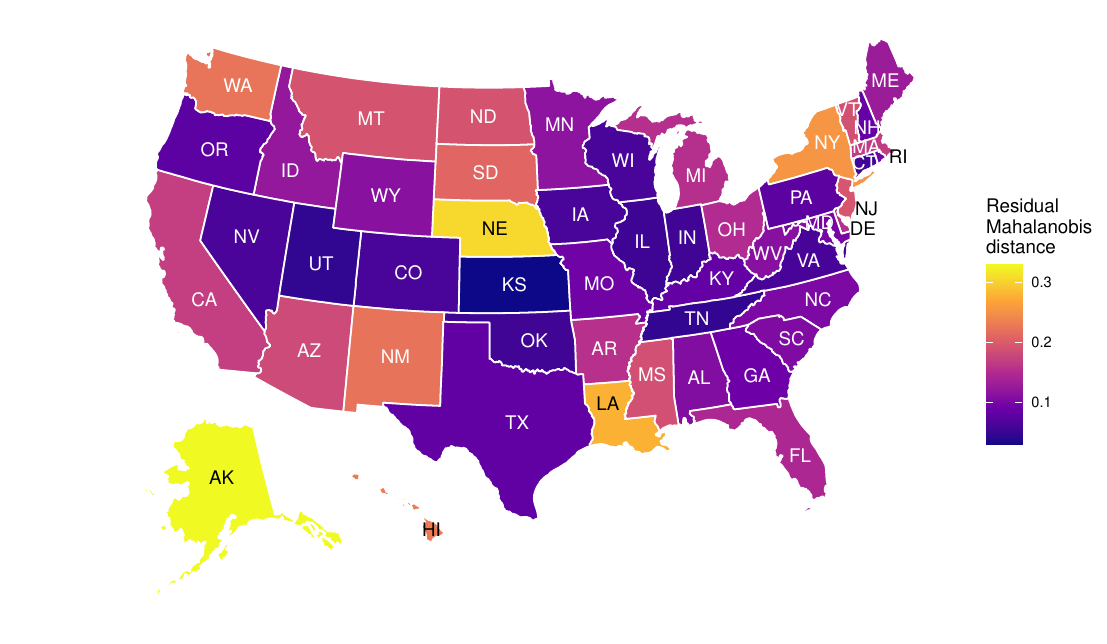}
  \caption{US map coloured by the residual Mahalanobis distances $D_n$ for each state $n=1,\ldots,50$.}
  \label{fig:US_map}
\end{figure}

\begin{figure}[]
  \centering
  \includegraphics[width=0.7\linewidth]{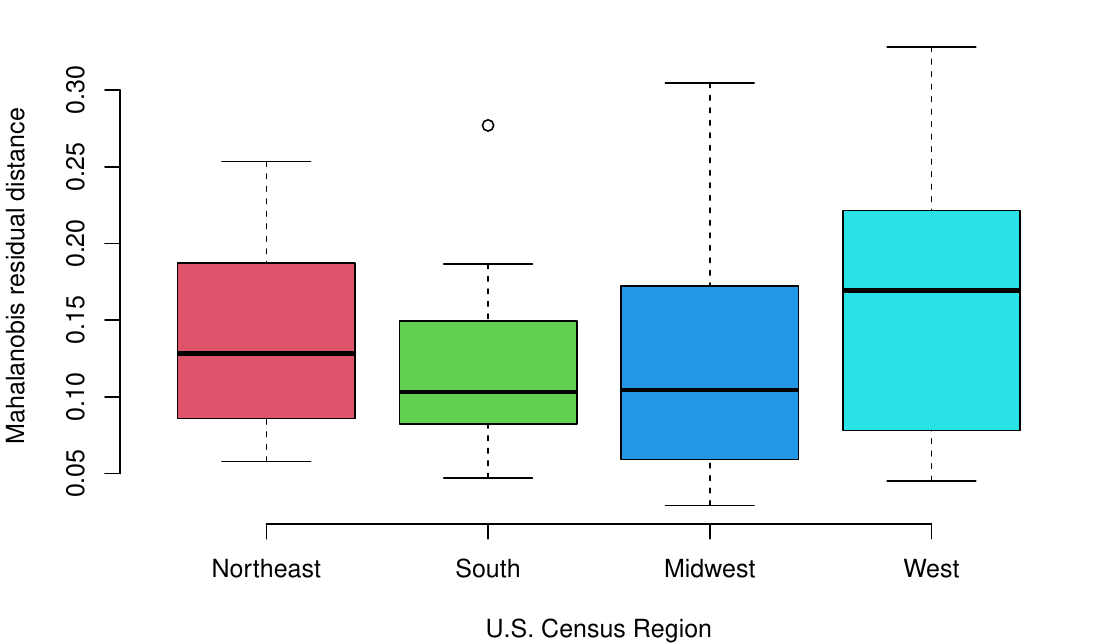}
  \caption{Residual Mahalanobis distances $D_n$ grouped by the four US Census regions.}
  \label{fig:boxplots}
\end{figure}

Finally, to assess the conditional linearity of the response-covariate relationships and to identify observations that may disproportionately determine individual regression coefficients, we examined added-variable plots based on the Frisch-Waugh-Lovell (FWL) theorem \citep{lovell1963}.
For each response and covariate, both the response and the selected covariate are first regressed on all but the selected covariate, and residuals from those two fits are plotted against one another.
By the FWL theorem, the slope of the least-squares line in this plot is exactly the estimated coefficient for the omitted variable in the full regression model.
The plots are thus useful not only to assess
whether linear associations are present, but also to examine whether the partial relationships are adequately represented in a linear manner. To this
end, Figure~\ref{fig:added_vars_FWL} shows such a partial association plot for each response (columns) and each variable (rows) in the model.
The added variable plots are broadly consistent with the
assumed linear relationships, and no single observation
appears to dominate the estimated coefficients.

\begin{figure}
  \centering
  \includegraphics[width=\linewidth]{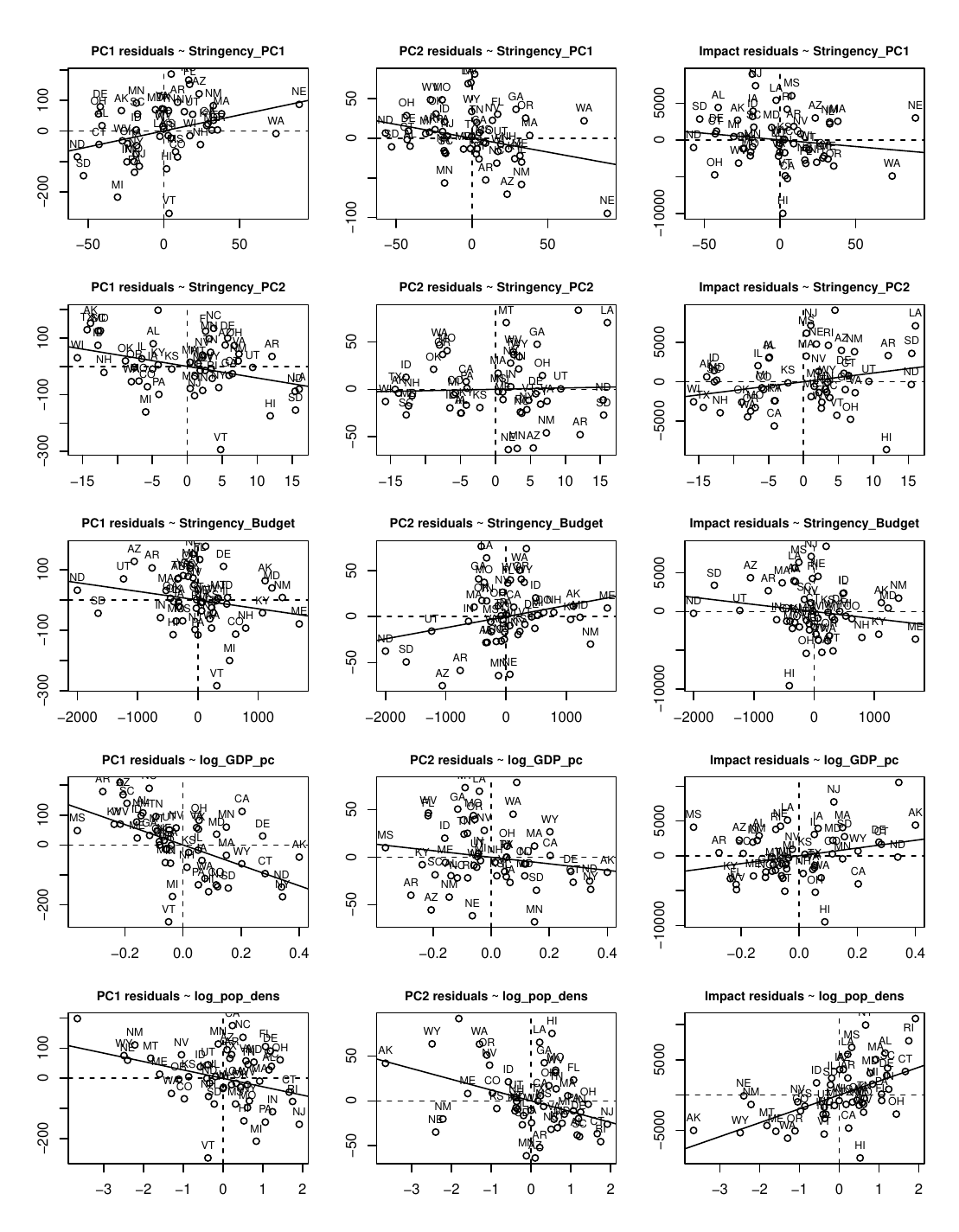}
  \caption{Added variable plots based on the FWL theorem. The solid lines denote the least-squares fits.}
  \label{fig:added_vars_FWL}
\end{figure}

\section{Proof of Proposition 1}\label{sup:proof}

\begin{proof}
Note that, writing $F^{-1}_\mu,\hat F^{-1}_\mu$ for the quantile functions corresponding to $\mu,\hat\mu$ respectively, we have that:
\begin{equation*}
  W_2(\hat\mu,\mu)=\|F^{-1}_\mu-F^{-1}_{\hat\mu}\|_{L^2((0,1),\lambda)}=\|\id-F^{-1}_{\hat\mu}\circ F_\mu\|_{L^2(\mathcal{X},\mu)}
\end{equation*}
for $\lambda$ the Lebesgue measure on $(0,1)$, the second equality being the change of variables $u=F_\mu(x)$. By the fixed-point characterisation of barycenters [9, 3.1.4]:
\begin{equation*}
  F^{-1}_{\hat\mu}(\cdot)=\frac1n\sum_{i=1}^nF^{-1}_{\nu_i}(\cdot),
\end{equation*}
which readily shows:
\begin{equation*}
  W_2(\hat\mu,\mu)=\Bigl\|\id-\frac1n\sum_{i=1}^nF^{-1}_{\nu_i}\circ F_\mu\Bigr\|_{L^2(\mathcal{X},\mu)}=\Bigl\|\id-\frac1n\sum_{i=1}^nT_i\Bigr\|_{L^2(\mathcal{X},\mu)},
\end{equation*}
where $T_i=F^{-1}_{\nu_i}\circ F_\mu$ because $T_i$ is increasing and $\mu$ is atomless, so that $T_i$ is the optimal map from $\mu$ to $\nu_i$. Since the $T_i-\id$ are i.i.d.\ and mean-zero in $L^2(\mathcal{X},\mu)$, this proves $W_2(\hat\mu,\mu)^2=O_p(n^{-1})$ by the central limit theorem in Hilbert spaces.

Let us define the operators
\begin{equation*}
  C:=\mathbb{E}\bigl[(T_1-\id)\otimes(T_1-\id)\bigr]\qquad\text{and}\qquad\hat C:=\frac1n\sum_{j=1}^n(\hat T_j-\id)\otimes(\hat T_j-\id),
\end{equation*}
where $T_j:=T^{\nu_j}_\mu$ and $\hat T_j:=T^{\nu_j}_{\hat\mu}$ for $j\ge1$. Next, we need to introduce the concept of parallel transport, where $P:L^2(\mathcal{X},\hat\mu)\to L^2(\mathcal{X},\mu)$, $Pv=v\circ F^{-1}_{\hat\mu}\circ F_\mu$, which is a surjective isometry mapping $L^2(\mathcal{X},\hat\mu)$ to $L^2(\mathcal{X},\mu)$ in a geometry preserving fashion, with $P^*=P^{-1}$. Then, it is easily seen that:
\begin{align*}
  \bigl\|C-P\hat CP^*\bigr\| &= \Bigl\|\mathbb{E}\bigl[(T_1-\id)\otimes(T_1-\id)\bigr]-\frac1n\sum_{j=1}^nP\bigl(\hat T_j-\id\bigr)\otimes P\bigl(\hat T_j-\id\bigr)\Bigr\|\\
  &\le O_p(n^{-1/2})+\Bigl\|\frac1n\sum_{j=1}^n(T_j-\id)\otimes(T_j-\id)-\frac1n\sum_{j=1}^nP\bigl(\hat T_j-\id\bigr)\otimes P\bigl(\hat T_j-\id\bigr)\Bigr\|\\
  &\le O_p(n^{-1/2})+\Bigl\|\frac1n\sum_{j=1}^n\bigl[(T_j-\id)-P(\hat T_j-\id)\bigr]\otimes(T_j-\id)\Bigr\|
  \\ & \qquad +\Bigl\|\frac1n\sum_{j=1}^nP\bigl(\hat T_j-\id\bigr)\otimes\bigl[(T_j-\id)-P(\hat T_j-\id)\bigr]\Bigr\|\\
  &= O_p(n^{-1/2})+\Bigl\|\frac1n\sum_{j=1}^n(\id-P\id)\otimes(T_j-\id)\Bigr\|+\Bigl\|\frac1n\sum_{j=1}^nP\bigl(\hat T_j-\id\bigr)\otimes(\id-P\id)\Bigr\|,
\end{align*}
where the first inequality is the law of large numbers for i.i.d.\ Hilbert--Schmidt operators, available since $\mathbb{E}\|T-\id\|^4\le|\mathcal{X}|^4<\infty$ by compactness, and where we have used that:
\begin{equation*}
  T_j-P\hat T_j=F^{-1}_{\nu_j}\circ F_\mu-\bigl(F^{-1}_{\nu_j}\circ F_{\hat\mu}\bigr)\circ F^{-1}_{\hat\mu}\circ F_\mu=F^{-1}_{\nu_j}\circ F_\mu-F^{-1}_{\nu_j}\circ F_\mu=0.
\end{equation*}
Since $\|\id-P\id\|_{L^2(\mathcal{X},\mu)}=\|\id-F^{-1}_{\hat\mu}\circ F_\mu\|_{L^2(\mathcal{X},\mu)}=O_p(n^{-1/2})$ by the first part, while $\|T_j-\id\|\le|\mathcal{X}|$ and $\|P(\hat T_j-\id)\|\le\|T_j-\id\|+\|\id-P\id\|$ are $O_p(1)$, this shows that:
\begin{equation}
  \label{eq:opbound}
  \bigl\|C-P\hat CP^*\bigr\|=O_p(n^{-1/2}).
\end{equation}

Note that, by the spectral decomposition theorem [3, Theorem 7.2.6], we may write $\hat C=\sum_{j\ge1}\hat\sigma_j^2\,\hat\phi_j\otimes\hat\phi_j$, for non-negative eigenvalues $\hat\sigma_j^2\ge0$ and orthonormal eigenfunctions $\hat\phi_j\in L^2(\mathcal{X},\hat\mu)$, $j\ge1$. Hence, we have that
\begin{equation*}
  P\hat CP^*=\sum_{j\ge1}\hat\sigma_j^2\,(P\hat\phi_j)\otimes(P\hat\phi_j)
\end{equation*}
and it is easy to see that $\{P\hat\phi_j\}_{j\ge1}$ form an orthonormal sequence in $L^2(\mathcal{X},\mu)$ of eigenfunctions for the operator $P\hat CP^*$, with corresponding eigenvalue sequence given by $\{\hat\sigma_j^2\}_{j\ge1}$. Hence, the usual perturbation argument gives
\begin{equation*}
  \sup_{1\le j\le R}\bigl|\hat\sigma_j^2-\sigma_j^2\bigr|^2=O_p(n^{-1}).
\end{equation*}

As for the result on principal components, we adapt the argument in Bosq [1, Lemma 4.3]. Write $e_j:=P\hat\phi_j$, and complete $\{\phi_j\}_{j\le R}$ to an orthonormal basis $\{\phi_\ell\}_{\ell\ge1}$ of $L^2(\mathcal{X},\mu)$ by adjoining an orthonormal basis of $\ker C$. First note that:
\begin{equation*}
  Ce_j-\sigma_j^2e_j=\bigl(C-P\hat CP^*\bigr)e_j+\bigl(\hat\sigma_j^2-\sigma_j^2\bigr)e_j
\end{equation*}
yielding $\|Ce_j-\sigma_j^2e_j\|=O_p(n^{-1/2})$ by \eqref{eq:opbound} and by the eigenvalue bound just obtained. On the other hand:
\begin{equation*}
  \|e_j-\phi_j\|^2=\sum_{\ell\ge1}\langle e_j-\phi_j,\phi_\ell\rangle^2=\sum_{\ell\ne j}\langle e_j,\phi_\ell\rangle^2+\bigl(1-\langle e_j,\phi_j\rangle\bigr)^2.
\end{equation*}
Moreover, since $\sum_{\ell\ge1}\langle e_j,\phi_\ell\rangle^2=1$ and $\langle e_j,\phi_j\rangle\in[0,1]$ by the choice of signs, we observe that:
\begin{equation*}
  \bigl(1-\langle e_j,\phi_j\rangle\bigr)^2\le\bigl(1-\langle e_j,\phi_j\rangle\bigr)\bigl(1+\langle e_j,\phi_j\rangle\bigr)=1-\langle e_j,\phi_j\rangle^2=\sum_{\ell\ne j}\langle e_j,\phi_\ell\rangle^2.
\end{equation*}
Therefore, we have that $\|e_j-\phi_j\|^2\le2\sum_{\ell\ne j}\langle e_j,\phi_\ell\rangle^2$, and hence:
\begin{equation*}
  O_p(n^{-1})=\bigl\|Ce_j-\sigma_j^2e_j\bigr\|^2=\sum_{\ell\ge1}\bigl(\sigma_\ell^2-\sigma_j^2\bigr)^2\langle e_j,\phi_\ell\rangle^2\ \ge\ \Delta^2\sum_{\ell\ne j}\langle e_j,\phi_\ell\rangle^2\ \ge\ \frac{\Delta^2}{2}\,\|e_j-\phi_j\|^2,
\end{equation*}
where $\Delta=\min\{\sigma_j^2-\sigma_{j+1}^2,\ j=1,\dots,R\}>0$ with the convention $\sigma_{R+1}^2:=0$, so that $|\sigma_\ell^2-\sigma_j^2|\ge\Delta$ for every $\ell\ne j$ and $j\le R$; in particular the terms with $\ell>R$, for which $\sigma_\ell^2=0$ and $|\sigma_\ell^2-\sigma_j^2|=\sigma_j^2\ge\sigma_R^2\ge\Delta$, must be retained rather than truncated.

Finally, we have:
\begin{equation*}
  \sum_{j=1}^R\langle\hat T_i-\id,\hat\phi_j\rangle_{L^2(\mathcal{X},\hat\mu)}\,P\hat\phi_j=\sum_{j=1}^R\langle T_i-\id,P\hat\phi_j\rangle_{L^2(\mathcal{X},\mu)}\,P\hat\phi_j+o_p(1)=\sum_{j=1}^R\langle T_i-\id,\phi_j\rangle_{L^2(\mathcal{X},\mu)}\,\phi_j+o_p(1),
\end{equation*}
the first equality by the isometry of $P$ together with $P(\hat T_i-\id)=(T_i-\id)+(\id-P\id)$ and $\|\id-P\id\|=O_p(n^{-1/2})$, the second by the convergence of the $P\hat\phi_j$ just established and $\|T_i-\id\|\le|\mathcal{X}|$. Since the variability of $T-\id$ is of rank $R$, we have $\mathbb{E}\|(T-\id)-\sum_{j\le R}\langle T-\id,\phi_j\rangle\phi_j\|^2=\sum_{j>R}\sigma_j^2=0$, so that the last sum equals $T_i-\id$ almost surely and $\exp_\mu(T_i-\id)=(T_i)_\#\mu=\nu_i$. The proof is completed by continuity of the exponential map, in the form $W_2\bigl(\exp_\mu(v),\exp_\mu(w)\bigr)\le\|v-w\|_{L^2(\mathcal{X},\mu)}$, which holds because $(\id+v,\id+w)_\#\mu$ is a coupling of $\exp_\mu(v)$ and $\exp_\mu(w)$.
\end{proof}

\end{document}